\documentclass{article}

\usepackage{amsfonts,dsfont}
\usepackage{bm}
\usepackage{mathtools}
\usepackage{stmaryrd,ulem}
\usepackage{a4wide}

\usepackage[round]{natbib}

\usepackage{hyperref}

\usepackage{amsthm}
\newtheorem{theorem}{Theorem}[section]
\newtheorem{proposition}[theorem]{Proposition}
\newtheorem{corollary}[theorem]{Corollary}
\newtheorem{lemma}[theorem]{Lemma}

\theoremstyle{remark}
\newtheorem{remark}[theorem]{Remark}

\theoremstyle{definition}

\usepackage[ruled]{algorithm2e}

\newcommand{\R}{\mathbb{R}}
\newcommand{\PI}{\mathcal{I}}

\newcommand{\AdaDetect}{\mbox{AD}}

\newcommand{\Cov}{\mbox{Cov}}
\renewcommand{\P}{\mathbb{P}}
\newcommand{\E}{\mathbb{E}}

\newcommand{\TDP}{\mathrm{TDP}}
\newcommand{\FCP}{\mathrm{FCP}}
\newcommand{\FCR}{\mathrm{FCR}}

\newcommand{\FDR}{\mathrm{FDR}}

\newcommand{\FDP}{\mathrm{FDP}}

\newcommand{\mbe}{\mathbb{E}}

\newcommand{\mbp}{\mathbb{P}}

\newcommand{\mtc}{\mathcal}
\newcommand{\mbf}{\mathbf}

\newcommand{\wt}[1]{{\widetilde{#1}}}
\newcommand{\wh}[1]{{\widehat{#1}}}
\newcommand{\ol}[1]{\overline{#1}}

\RequirePackage{xstring}

\let\originalleft\left
\let\originalright\right
\renewcommand{\left}{\mathopen{}\mathclose\bgroup\originalleft}
\renewcommand{\right}{\aftergroup\egroup\originalright}

\newcommand{\paren}[2][a]{%
\IfEqCase{#1}{%
{a}{\left(#2\right)}%
{0}{(#2)}%
{1}{\big(#2\big)}%
{2}{\Big(#2\Big)}%
{3}{\bigg(#2\bigg)}%
{4}{\Bigg(#2\Bigg)}%
}[\PackageError{paren}{Undefined option to paren: #1}{}]%
}
\newcommand{\norm}[2][a]{%
\IfEqCase{#1}{%
{a}{\left\lVert#2\right\rVert}%
{0}{\lVert#2\rVert}%
{1}{\big\lVert#2\big\rVert}%
{2}{\Big\lVert#2\Big\rVert}%
{3}{\bigg\lVert#2\bigg\rVert}%
{4}{\Bigg\lVert#2\Bigg\rVert}%
}[\PackageError{norm}{Undefined option to norm: #1}{}]%
}
\newcommand{\brac}[2][a]{%
\IfEqCase{#1}{%
{a}{\left[#2\right]}%
{0}{[#2]}%
{1}{\big[#2\big]}%
{2}{\Big[#2\Big]}%
{3}{\bigg[#2\bigg]}%
{4}{\Bigg[#2\Bigg]}%
}[\PackageError{brac}{Undefined option to brac: #1}{}]%
}
\newcommand{\inner}[2][a]{%
\IfEqCase{#1}{%
{a}{\left\langle#2\right\rangle}%
{0}{\langle#2\rangle}%
{1}{\big\langle#2\big\rangle}%
{2}{\Big\langle#2\Big\rangle}%
{3}{\bigg\langle#2\bigg\rangle}%
{4}{\Bigg\langle#2\Bigg\rangle}%
}[\PackageError{inner}{Undefined option to inner: #1}{}]%
}
\newcommand{\abs}[2][a]{
\IfEqCase{#1}{%
{a}{\left\vert#2\right\rvert}%
{0}{\vert#2\rvert}%
{1}{\big\vert#2\big\rvert}%
{2}{\Big\vert#2\Big\rvert}%
{3}{\bigg\vert#2\bigg\rvert}%
{4}{\Bigg\vert#2\Bigg\rvert}%
}[\PackageError{abs}{Undefined option to abs: #1}{}]%
}

\newcommand{\set}[2][a]{
\IfEqCase{#1}{%
{a}{\left\{#2\right\}}%
{0}{\{#2\}}%
{1}{\big\{#2\big\}}%
{2}{\Big\{#2\Big\}}%
{3}{\bigg\{#2\bigg\}}%
{4}{\Bigg\{#2\Bigg\}}%
}[\PackageError{set}{Undefined option to set: #1}{}]%
}

\newcommand{\ind}[2][a]{{\mbf{1}\set[#1]{#2}}}

\newcommand{\e}[2][a]{\mbe\brac[#1]{#2}}

\newcommand{\prob}[2][a]{\mbp\brac[#1]{#2}}

\newcommand{\cH}{{\mtc{H}}}

\renewcommand{\l}{\ell}

\usepackage[linecolor=blue!60!,backgroundcolor=blue!10!,textwidth=1.8cm,textsize=tiny]{todonotes}
\newcommand{\stkout}[1]{\ifmmode\text{\sout{\ensuremath{#1}}}\else\sout{#1}\fi}

\definecolor{mygreen}{rgb}{0.82, 1.0, 0.82}
\definecolor{myred}{rgb}{ 1.0, 0.84, 0.84}

\newcommand{\range}[1]{\llbracket #1\rrbracket}

\newcommand{\dcal}{\mathcal{D}_{{\tiny \mbox{cal}}}}
\newcommand{\dtest}{\mathcal{D}_{{\tiny \mbox{test}}}}
\newcommand{\dtrain}{\mathcal{D}_{{\tiny \mbox{train}}}}
\newcommand{\dxct}{\mathcal{D}^{X}_{{\tiny \mbox{cal+test}}}}

\begin{document}

\title{Transductive conformal inference with adaptive scores}

\author{Ulysse Gazin\footnote{Universit\'e Paris Cit\'e and Sorbonne Universit\'e, CNRS, Laboratoire de Probabilit\'es, Statistique et Mod\'elisation. Email: ugazin@lpsm.paris} $\quad$ Gilles Blanchard\footnote{Universit\'e Paris Saclay, Institut Math\'ematique d'Orsay. Email: gilles.blanchard@universite-paris-saclay.fr} $\quad$ Etienne Roquain\footnote{Sorbonne Universit\'e and Universit\'e Paris Cit\'e, CNRS, Laboratoire de Probabilit\'es, Statistique et Mod\'elisation. Email: etienne.roquain@upmc.fr}}

  \maketitle

\bigskip

\begin{abstract}
  Conformal inference is a fundamental and versatile tool that provides distribution-free guarantees for many machine learning tasks. We consider the transductive setting, where decisions are made on a test sample of $m$ new points, giving rise to 
  $m$ conformal $p$-values. {While classical results only concern their marginal distribution, we show that their joint distribution
 follows   a P\'olya urn model, 
and establish
a concentration inequality for their empirical distribution function.} The results hold for arbitrary exchangeable scores, including {\it adaptive} ones that can use the covariates of the test+calibration samples at training stage for increased accuracy. We demonstrate the usefulness of these  theoretical
results through uniform, in-probability guarantees for two machine learning tasks of current interest: interval prediction for transductive transfer learning and novelty detection based on two-class classification.
\end{abstract}

\section{Introduction}\label{sec:intro}

Conformal inference is a general framework aiming at providing sharp uncertainty quantification guarantees
for the output of machine learning algorithms used as ``black boxes''.
A central tool of that field is the construction of a ``(non)-conformity score'' $S_i$ for each sample point.
The score functions can be learnt on a training set using various machine learning
methods depending on the task at hand.
The scores observed 
on a data sample called ``calibration sample'' $\dcal$ serve as references for the scores of a ``test sample'' $\dtest$ (which may or may not be observed, depending on the setting). The central property of these scores is that they are an exchangeable family of random variables.

\subsection{Motivating tasks}

To be more concrete, we start with two specific settings serving both as motivation and
as application. 

\begin{itemize}
\item[(PI)] Prediction intervals: we observe $\mathcal{D}_{{\tiny \mbox{cal}}}=(X_1,Y_1), \dots, (X_n,Y_n)$ a sample of i.i.d. variables with unknown distribution $P$, where $X_i\in \R^d$ is a regression covariate and $Y_i\in \R$ is the outcome. Given a new independent datum $(X_{n+1},Y_{n+1})$ generated from $P$, {the task is to} build a 
  prediction interval for $Y_{n+1}$ given $X_{n+1}$ and $\mathcal{D}_{{\tiny \mbox{cal}}}$.
  More generally, in the {\it transductive} conformal setting \citep{vovk2013transductive}, the task is
  repeated $m\geq 1$ times: 
  given $m$ new data points $\mathcal{D}_{{\tiny \mbox{test}}}=(X_{n+1},Y_{n+1}), \dots, (X_{n+m},Y_{n+m})$ i.i.d from $P$, build $m$ 
  prediction intervals for $Y_{n+1},\dots,Y_{n+m}$ given $X_{n+1},\dots,X_{n+m}$ and $\mathcal{D}_{{\tiny \mbox{cal}}}$. 
\item[(ND)] Novelty detection: we observe  $\mathcal{D}_{{\tiny \mbox{cal}}}=(X_1,\dots,X_n)$, a sample of nominal data points in $\R^d$, drawn i.i.d. from an unknown (null) distribution $P_0$, and a 
  test sample $\mathcal{D}_{{\tiny \mbox{test}}}=(X_{n+1},\dots,X_{n+m})$ of independent points in $\R^d$, each of which is distributed as $P_0$ or not.  {The task is to}  decide if each $X_{n+i}$ is distributed as the training sample (i.e., from $P_0$) or is a ``novelty''. 
\end{itemize}

\begin{figure}[h!]
\vspace{-1cm}
\center
\hspace{-8mm}
\includegraphics[scale=0.38]{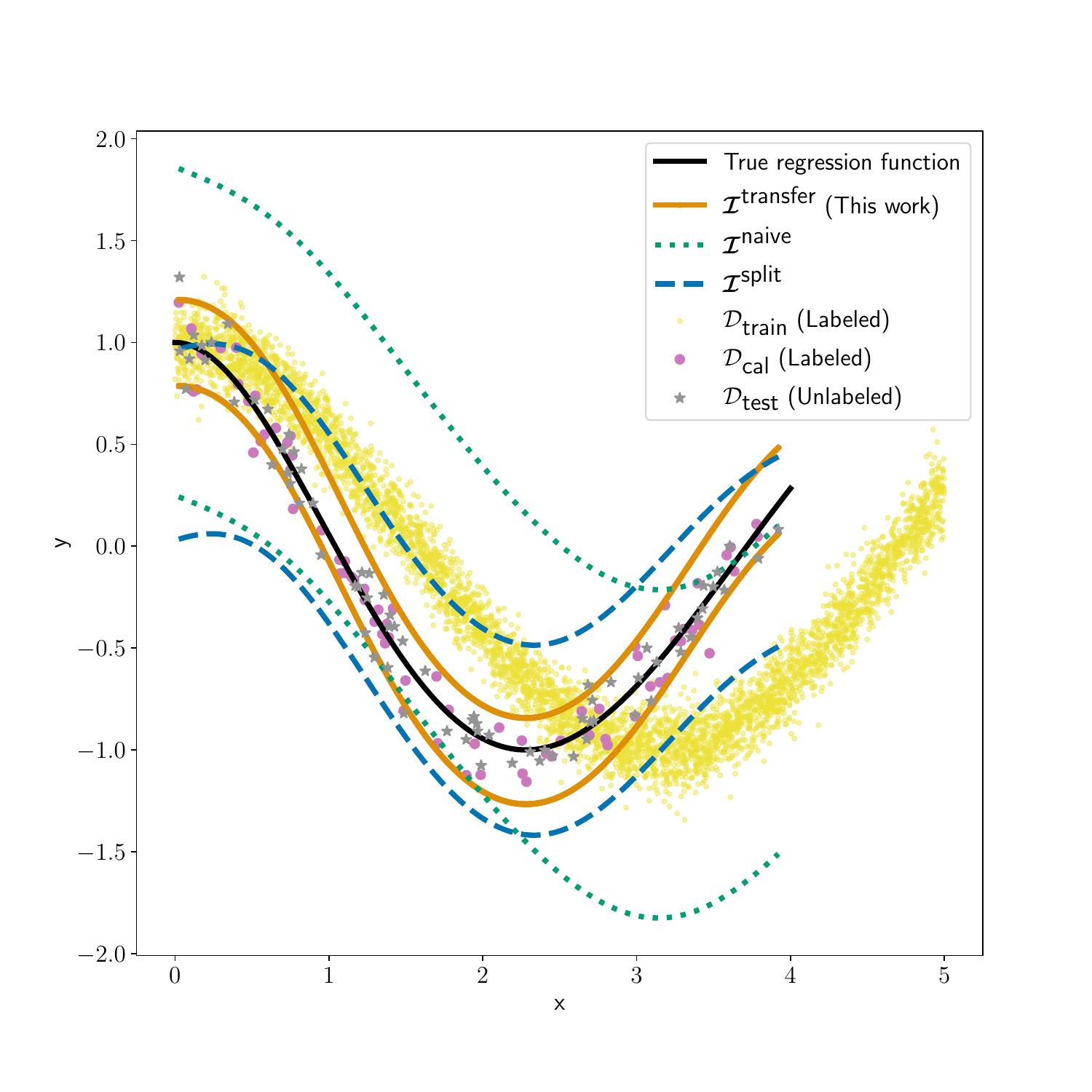}
\vspace{-1cm}
\caption{Task (PI) with adaptive scores in a non-parametric regression setting with domain shift between train and calibration+test samples
  ( 
{proof-of-concept model,} see Section~\ref{sec:numexpPI}).  Our contribution is both to propose adaptive scores and predictions relying on transfer learning (this figure), and uniform bounds on the false coverage proportion, see Figure~\ref{fig:IlluAlpha_L}. 
  \label{fig:IlluTransfert}
  }
\end{figure}

For both inference tasks, the usual pipeline is based on the construction of non-conformity real-valued scores $S_1,\dots,S_{n+m}$ for each member of $\mathcal{D}_{{\tiny \mbox{cal}}}\cup\mathcal{D}_{{\tiny \mbox{test}}}$, which requires an additional independent training sample $\mathcal{D}_{{\tiny \mbox{train}}}$ (in the so-called ``split conformal'' approach):
\begin{itemize}
\item[(PI)] the scores are (for instance) the regression residuals $S_i=|Y_i-\mu(X_i;\mathcal{D}_{{\tiny \mbox{train}}})|$, $1\leq i\leq n+m$, where the function $\mu(x;\mathcal{D}_{{\tiny \mbox{train}}})$ is a point prediction of $Y_i$ given $X_i=x$, learnt from the sample $\mathcal{D}_{{\tiny \mbox{train}}}$.
\item[(ND)] the scores are of the form $S_i=g(X_i;\mathcal{D}_{{\tiny \mbox{train}}})$, $1\leq i\leq n+m$, where the score function $g(\cdot;\mathcal{D}_{{\tiny \mbox{train}}})$ is learnt using the sample $\mathcal{D}_{{\tiny \mbox{train}}}$; $g(x)$ is meant to be large if $x$ is fairly different from the members of $\mathcal{D}_{{\tiny \mbox{train}}}$ (so that it is ``not likely'' to have been generated from $P_0$).  
\end{itemize}
In both cases, 
inference is based on the so-called split conformal $p$-values \citep{vovk2005algorithmic}:
\begin{equation}\label{equemppvalues}
p_i=(n+1)^{-1}\paren[3]{1+\sum_{j=1}^n \ind{S_j\geq S_{n+i}}},\:\: i\in \range{m}.
\end{equation} 
In other words, $(n+1)p_i$ is equal to the rank of $S_{n+i}$ in the set of values $\{S_1,\dots, S_n,S_{n+i}\}$,
and a small $p$-value $p_i$ indicates that the test score $S_{n+i}$ is abnormally high within the set of reference scores.
The link 
to the two above 
tasks is as follows: for (PI), the prediction interval {$\bm{\mathcal{C}}(\alpha)$ for $Y_{n+i}$ with coverage probability $(1-\alpha)$ is obtained by inverting the inequality $p_i>\alpha$ w.r.t. $Y_{n+i}$, see \eqref{equrules} below.} For (ND), the members of the test sample declared as novelties are those with a $p$-value $p_i\leq t$ for some 
threshold $t$.

Studying the behavior of the conformal $p$-value family is thus a cornerstone of 
conformal inference. Still, classical results only concern the {\it marginal} distribution of the $p$-values while the joint distribution remains largely unexplored in full generality.

\subsection{Contributions and overview of the paper}

In Section~\ref{sec:main}, we present new results for the joint distribution of the conformal $p$-values  \eqref{equemppvalues} for general exchangeable scores (for any sample sizes $n$ and $m$). First, in Section~\ref{sec:key}, we show that the dependence structure involved {only depends on $n$ and $m$, and} 
follows a P\'olya urn model; this entails both explicit formula and useful characterizations. Second,
we deduce a new finite sample DKW-type concentration inequality~\citep{Mass1990} for the empirical distribution function (ecdf) of the conformal $p$-values.

We illustrate the interest of the theoretical results through the application cases (PI) in Section~\ref{sec:PI} and (ND) in Section~\ref{sec:ND}, for which dedicated numerical experiments are also provided\footnote{The code used in all our experiments is made publicly available at \url{https://github.com/ulyssegazin/TransductiveAdaptive_CP}.}.

Our theory provides {\it in-probability {(i.e. confidence)} bounds} for the error proportion when $m$ decisions are taken simultaneously (transductive setting);
furthermore, these bounds are {\it uniform} over a certain class of decisions. {More precisely}, the proportion of errors among the $m$ decisions corresponds to 
the  false coverage proportion (FCP) for (PI), resp. the false discovery proportion (FDP) for (ND). We develop upper confidence bounds for these
quantities, {in dependence of prediction interval length for (PI), resp. the rejection threshold for (ND),
  and valid {\it uniformly} over the choice of these parameters.}
{This is in contrast to marginal guarantees in previous literature only providing in-expectation guarantees of FCP/FDP
  at a fixed level $\alpha$, and for specific procedures.}
{Obtaining in-probability bounds for the FDP is a classical and active theme of multiple testing theory: 
  in contrast to an in-expectation control, it takes into account the fluctuations of the error proportion. It thus brings more fine-grained reliability, while the uniform guarantee also offers the practitioner more flexibility
  for taking a data-driven decision that is still theoretically backed up.
  These guarantees can in particular be crucial when handling sensible data.}
  {Similarly, obtaining a sharp confidence bound on the actual (random) number of false inferences for (PI) for
  repeated decisions
  is much more informative than a bound on its expectation.}

We insist that we only assume that the scores are exchangeable to obtain in-probability guarantees.
Exchangeability is a classical assumption in conformal theory,
though some recent works have sometimes dropped it in favor of i.i.d. scores.
Deriving results under the weaker exchangeable assumption is crucial in the considered applications,
because while the {\it data} is assumed i.i.d., we rely on {\it adaptive conformal scores} which depend not only on the training sample (arbitrarily), but also {{\it on the calibration+test sample} in an exchangeable way.
  Adaptive scores offer superior performance in practice (see Figure \ref{fig:IlluTransfert} for our approach to transductive transfer PI), are indeed exchangeable (thus, our theory applies) but {\it not} i.i.d. This illustrates the interest to develop the joint distribution theory under the weaker exchangeable assumption 
  {\it even} if the underlying data is assumed to be i.i.d., a standard setting (our results also hold if the data is only assumed exchangeable).

\subsection{Relation to previous work}\label{sec:rel}

For fundamentals on conformal prediction, see \cite{vovk2005algorithmic,balasubramanian2014conformal}.
We only {consider} the {\it split conformal} approach, also named inductive conformal approach in the seminal work of \cite{papadopoulos2002inductive}.
The split conformal approach uses a separate training set but  is considered the
most practically amenable approach for big data (in contrast to the ``full conformal'' approach which can be sharper but computationally intractable).

The most important consequence of score exchangeability is that the marginal distribution of a conformal $p$-value is a discrete uniform under the joint (calibration and test) data distribution.
There has been  significant recent interest for the {\it conditional} distribution of a marginal $p$-value,
conditional to the calibration sample, under the stronger assumption of i.i.d. scores. The corresponding results take the form
of bounds on $\P(p_1\leq t \:|\: \dcal)$ holding with high probability over $\dcal$ (\citealp{vovk2012conditional,bian2022training,sarkar2023post,bates2023testing}, 
where in the two latter references the results are in addition uniformly valid in $t$).  
However, the i.i.d. scores assumption prevents handling adaptive scores
, for which only exchangeability is guaranteed; moreover, these works only handle a single prediction. 

Simultaneous inference for the (PI) task has been  proposed by \cite{vovk2013transductive} (see also \citealp{saunders1999transduction} for an earlier occurrence for one $p$-value with multiple new examples), referred to as transductive conformal inference. It includes a bound on the family-wise error rate
(the probability of committing one or more PI errors) based on a Bonferroni-type correction. In the present work we allow the number of PI errors to be positive but aim at a tight control of this number in probability (uniformly valid over the choice of PI length).

  Closest to our work, \cite{f2023universal,huang2023uncertainty}  analyze the false coverage proportion (FCP) of the usual prediction interval family $\bm{\mathcal{C}}(\alpha)$ repeated over $m$ test points: 
  the exact distribution of the FCP
  under data exchangeability is provided, and related in \cite{f2023universal} to a P\`olya urn model with two colors.
  We show the more general result that the {\it full joint} distribution of $(p_1,\dots,p_m)$ follows a P\`olya urn model with $(n+1)$ colors, which entails the result of \cite{f2023universal,huang2023uncertainty} as a corollary (see Supplemental~\ref{sec:polya}). This brings substantial innovations: our bounds on FCP are {\it uniform} in $\alpha$, and we provide both
  the exact joint distribution and an explicit non-asymptotic approximation via a DKW-type concentration bound.
  {Finally, \cite{bao2023selective} also established an in-expectation control 
    of the FCP after a selection stage.  
    By contrast, we provide FCP bounds in probability. In addition, while our theory is stated without selection stage, it can be applied to a permutation invariant selection, see Remark~\ref{rem:selection}.}

  The (ND) setting is alternatively referred to as Conformal Anomaly Detection (see Chapter 4 of \citealp{balasubramanian2014conformal}).
  We specifically consider here the (transductive) setting of \cite{bates2023testing} where the test sample contains novelties, and the corresponding $p$-values for `novelty' entries are not discrete uniform  but expected to be stochastically smaller.
  Due to strong connections to multiple testing,  ideas and procedures stemming from that area can be adapted to address (ND), specifically by controlling the false discovery rate (FDR, the expectation of the FDP), such as as the Benjamini-Hochberg (BH) procedure 
  \citep{BH1995}. 
  {Use of adaptive scores and corresponding} FDR control has been investigated by \cite{marandon2022machine}. 
  Our contribution with respect to that work comes from getting uniform and in-probability bounds for the FDP
  (rather than only in expectation, for the FDR).

\section{Main results}\label{sec:main}

\subsection{Setting}\label{sec:setting}

{We denote integer ranges using $\range{i}=\{1,\dots,i\}$, $\range{i,j}=\{i,\dots,j\}$.}
Let $(S_i)_{i \in \range{n+m}}$ be real random variables corresponding to non-conformity scores, for which $(S_j)_{j \in \range{n}}$ are the ``reference'' scores and $(S_{n+i})_{i \in \range{m}}$ are the ``test'' scores. 
We assume
\begin{equation}\label{as:exchangeable}\tag{Exch}
\mbox{The score vector $(S_i)_{i\in \range{n+m}}$ is exchangeable.}
\end{equation}

Under \eqref{as:exchangeable}, the $p$-values \eqref{equemppvalues} have super-uniform marginals (see, e.g., \citealp{RW2005}).
In addition, the marginal distributions are all equal and uniformly distributed on $\{\ell/(n+1),\ell\in\range{n+1}\}$ under the additional mild assumption:
\begin{equation}\label{as:noties}\tag{NoTies}
\mbox{The score vector $(S_i)_{i\in \range{n+m}}$ has no ties a.s.}
\end{equation}

While the marginal distribution is well identified, the joint distribution of the $p$-values is not well studied yet. In particular, we will be interested in the 
empirical distribution function of the $p$-value family, defined as
\begin{equation}\label{equ-ecdfpvalues}
\wh{F}_m(t):=m^{-1}\sum_{i=1}^m \ind{p_i\leq t}, \:\:\: t\in [0,1].
\end{equation}
Note that the $p$-values are not i.i.d. under \eqref{as:exchangeable}, so that most classical concentration inequalities, such as DKW's inequality \citep{Mass1990}, or Bernstein's inequality, cannot be directly used. Instead, we should take into account the specific dependence structure underlying these $p$-values.

\subsection{Key properties}\label{sec:key}

We start with  a straightforward result, under the stronger assumption
\begin{equation}\label{as:iid}\tag{IID}
\mbox{The variables $S_i ,i\in \range{n+m}$, are i.i.d.}
\end{equation}
For this, introduce, for any fixed vector $U=(U_1,\dots,U_n)\in [0,1]^n$, the discrete distribution $P^U$ on the set $\set[1]{\frac{\ell}{n+1}, \ell \in \range{n+1}}$,  defined as 
\begin{equation}\label{equPU}
P^U(\{\l/(n+1)\})=U_{(\l)}-U_{(\l-1)}, \:\:\:\l\in \range{n+1},
\end{equation}
where $0=U_{(0)}\leq U_{(1)}\leq \dots \leq U_{(n)}\leq U_{(n+1)}=1$  are the increasingly ordered values of $U=(U_1,\dots,U_n)$. {In words, the $n$ values of $U$ divide the interval $[0,1]$ into $(n+1)$ distinct cells
  (labeled $\frac{\ell}{n+1}, \ell\in \range{n+1}$), and $P^U$ is
  the probability distribution of the label of the cell
  a $\mathrm{Unif}[0,1]$ variable would fall into.}

Note that $P^U$ has for c.d.f. 
\begin{equation}\label{equFU}
F^U(x)=U_{(\lfloor (n+1)x\rfloor)}, \:\:\:x\in [0,1].
\end{equation}

{The following result can be considered as well known from previous literature (see, e.g.,  proof of Theorem~6 in \citealp{bates2023testing}); we include a short proof for completeness.}
\begin{proposition}\label{prop:iid}
  Assume \eqref{as:iid} and \eqref{as:noties} and consider the $p$-values $(p_i,i\in \range{m})$ given by \eqref{equemppvalues}. Then conditionally on  $\mathcal{D}_{{\tiny \mbox{cal}}}=(S_1,\dots,S_n)$, the $p$-values 
  are i.i.d. of common distribution given by 
$$
p_1 \:|\: \mathcal{D}_{{\tiny \mbox{cal}}} \sim P^U,
$$
where  $U=(U_1,\dots,U_n)=\paren[1]{1-F(S_1),\dots,1-F(S_n)}$ are pseudo-scores
and $F$ is the common c.d.f. of the scores of $\mathcal{D}_{{\tiny \mbox{cal}}}$, that is, $F(s)=\P(S_1\leq s)$, $s\in \R$.
In addition the pseudo-score vector $U$ is i.i.d. $\mathrm{Unif}[0,1]$ distributed.
\end{proposition}
{{\bf Proof sketch.} 
  The conditional distribution
  of $p_i$ only depends on score ordering which is unambiguous due to \eqref{as:noties},
  and is thus invariant by monotone transformation of the scores by $(1-F)$.
  Writing explicitly the cdf of $p_i$ from the uniformly distributed transformed scores yields~\eqref{equFU}.
  See Supplemental~\ref{proofprop:iid} for details.}

In the literature, such a result is used  to control the conditional failure  probability $\P(p_1\leq \alpha \:|\: \mathcal{D}_{{\tiny \mbox{cal}}})$ around its expectation (which is ensured to be smaller than, and close to, $\alpha$) with concentration inequalities valid under an i.i.d. assumption \citep{bates2023testing, sarkar2023post, bian2023training}.

By integration over $U$, a direct consequence of Proposition~\ref{prop:iid} is that, under \eqref{as:iid} and \eqref{as:noties}, and {\it unconditionally on $\mathcal{D}_{{\tiny \mbox{cal}}}$}, the family of conformal $p$-values $(p_i,i\in \range{m})$ has the ``universal'' distribution 
$P_{n,m}$ on $[0,1]^m$ defined as follows:
\begin{align}
P_{n,m} = \mathcal{D}\big(q_i,i\in \range{m}\big)
  \label{equ:distribution}, \text{ where }\\
  \begin{cases}
    \hfill\big(q_1,\ldots,q_m |U\big) &\stackrel{\text{i.i.d.}}{\sim} P^U;\\
  \text{ and } U=(U_1,\ldots,U_n) &\stackrel{\text{i.i.d.}}{\sim} \mathrm{Unif}([0,1]).
  \end{cases} \label{equ:generation}
\end{align}

Our first result is to note that the latter holds beyond the i.i.d. assumption.

\begin{proposition}\label{prop:exch}
  Assume \eqref{as:exchangeable} and \eqref{as:noties}, then the family of $p$-values $(p_i,i\in \range{m})$ given by \eqref{equemppvalues}
  has joint distribution $P_{n,m}$, which is
  defined by \eqref{equ:distribution}-\eqref{equ:generation} and is independent of the specific score distribution.
\end{proposition}
{{\bf Proof sketch.} The joint distribution of the $p$-values only depends on the ranks of the $(n+m)$
  scores. Since the scores have exchangeable distribution and~\eqref{as:noties} holds, their ranks form a random permutation of $\range{n+m}$.
Thus, 
the same rank distribution (and consequently joint $p$-value distribution) is generated when the scores are i.i.d.
Applying Proposition~\ref{prop:iid}, the $p$-value distribution can be represented as~\eqref{equ:distribution}-\eqref{equ:generation}.
See also Supplemental~\ref{proofprop:exch}.}

{The next proposition is an alternative and useful characterization of the distribution $P_{n,m}$.}

\begin{proposition}\label{prop:polya}
  $P_{n,m}$ 
  is the distribution of
  the colors of $m$ successive draws in a standard P\'olya urn model with $n+1$ colors labeled $\big\{\frac{\l}{n+1},
    \l \in \range{n+1}\big\}$.
\end{proposition}

Proposition~\ref{prop:polya} is proved in Supplemental~\ref{sec:polya}, where several explicit formulas for $P_{n,m}$ are also provided. We also show that this generalizes the previous works of \cite{f2023universal,huang2023uncertainty}.

Comparing Proposition~\ref{prop:iid} and Proposition~\ref{prop:exch}, we see that having i.i.d. scores is more favorable because guarantees are valid conditionally on $\mathcal{D}_{{\tiny \mbox{cal}}}$ {(with an explicit expression for $U=U(\dcal)$)}. However, as we will see in Sections~\ref{sec:PI} and~\ref{sec:ND}, the class of exchangeable scores is much more flexible and includes adaptive scores, which can improve substantially inference sharpness in specific situations. For this reason, we work with the unconditional distribution as in Proposition~\ref{prop:exch} in the sequel.

\subsection{Consequences}
\label{sec:DKW}

We now provide a DKW-type envelope for 
the empirical distribution function~\eqref{equ-ecdfpvalues} of conformal $p$-values.
Let us introduce the discretized identity function
\begin{align}
I_n(t) =\lfloor (n+1)t\rfloor/(n+1)=\E \wh{F}_m(t),\:\:\:t\in [0,1],\label{discretizedid}
\end{align}
 and the following bound:
\begin{align}
B^{\mbox{\tiny DKW}}(\lambda,n,m)&:=
 \bm{1}_{\set{\lambda <1}} \left[1+\frac{2\sqrt{2\pi}\lambda \tau_{n,m}}{(n+m)^{1/2}}\right]e^{-2\tau_{n,m} \lambda^2},
\label{BDKW}
\end{align}
where $\tau_{n,m}:=nm/(n+m)\in [(n\wedge m)/2, n\wedge m]$ is an 
{``effective sample size''}.

\begin{theorem}\label{thDKW} Let us consider the process $\wh{F}_m$ defined by \eqref{equ-ecdfpvalues}, the discrete identity function $I_n(t)$ defined by \eqref{discretizedid}, and assume \eqref{as:exchangeable} and \eqref{as:noties}. Then we have  for all $\lambda>0$, $n,m\geq 1$,
\begin{align}
\P\Big(\sup_{t\in [0,1]}(\wh{F}_m(t) - I_n(t)) > \lambda\Big)&\leq B^{\mbox{\tiny DKW}}(\lambda,n,m).\label{boundDKWup}
\end{align}
In addition, $B^{\mbox{\tiny DKW}}(\lambda^{\mbox{\tiny DKW}}_{\delta,n,m},n,m)\leq \delta$ for
\begin{align}
&\lambda^{\mbox{\tiny DKW}}_{\delta,n,m}=\Psi^{(r)}(1);
\label{boundDKWupexplicit}\\
&\Psi(x)=1\wedge\bigg( \frac{\log(1/\delta)+\log\big(1+ \sqrt{2\pi}  \frac{2\tau_{n,m}x}{(n+m)^{1/2}}\big)}{2\tau_{n,m}} \bigg)^{1/2}, \notag
\end{align}
where $\Psi^{(r)}$ denotes the function $\Psi$ iterated $r$ times (for an arbitrary integer $r\geq 1$). 
\end{theorem}
{{\bf Proof sketch.} Use the representation~\eqref{equ:generation}, apply the DKW inequality separately
  to $(U_1,\ldots,U_n)$ and to $(q_1,\ldots,q_m)$ conditional to $U$, and integrate over $U$.}
See supplemental Section~\ref{sec:thDKWproof} for details (a slightly more accurate bound is also proposed). 

{In supplemental Section~\ref{sec:tighnessDKW}, we illustrate the sharpness of the inequality \eqref{boundDKWup}.}

\begin{remark}\label{remSimes}
The Simes inequality \citep{Sim1986} is true for conformal $p$-values \citep{bates2023testing}, which provides a different confidence envelope on $\wh{F}_m$. A comparison with the new DKW bound \eqref{BDKW} is provided in Supplemental~\ref{sec:Simes}. It shows that the latter is sharper in a wide range of situations. 
\end{remark}

\begin{remark}\label{numispossible}
Since the distribution $P_{n,m}$ can be easily sampled from, $\lambda^{\mbox{\tiny DKW}}_{\delta,n,m}$ in \eqref{boundDKWupexplicit} 
can be further improved by considering the sharper but implicit quantile
\begin{align*}
&\lambda^{\mbox{\tiny num-DKW}}(\delta,n,m)=\min\Big\{x\geq 0\::\: \pi_{n,m,x}\leq \delta\Big\}, \text{with}\\
&\pi_{n,m,x}:= 
  P_{n,m}\bigg(\sup_{\l\in \range{n+1}}\Big(\wh{F}_m\Big(\frac{\ell}{n+1}\Big) - \frac{\ell}{n+1}\Big) > x\bigg).\nonumber
\end{align*}
In addition, numerical confidence envelopes for $\wh{F}_m$ with other shapes can be investigated.   
For instance, for any set $\mathcal{K} \subset \range{m}$ of size $K$, we can calibrate thresholds $t_1,\dots,t_K>0$ such that 
\begin{align}
&\P_{\bm{p}\sim P_{n,m}}(\forall k\in \mathcal{K},\:p_{(k+1)}> t_k)\nonumber\\
 &= \P_{\bm{p}\sim P_{n,m}}(\forall k\in \mathcal{K},\: \wh{F}_m(t_k)\leq k/m)\geq 1-\delta.\label{exactcontrol}
\end{align}
A method is to start from a ``template'' one-parameter family $(t_k(\lambda))_{k\in \mathcal{K}}$ and then adjust $\lambda$ to obtain the desired control \citep{BNR2020,li2022simultaneous}. This approach is developed in detail in Suppl.~\ref{sec:template}.
\end{remark}

\section{Application to prediction intervals}\label{sec:PI}

In this section, we apply our results to build simultaneous conformal prediction intervals, with an angle towards adaptive scores and transfer learning.

\subsection{Setting}\label{sec:settingconformal}

Let us consider a conformal prediction framework for a regression task, see, e.g., \cite{lei2018distribution}, with three independent samples of points $(X_i,Y_i)$, where $X_i\in \R^d$ is the covariable and $Y_i\in \R$ is the outcome:
\begin{itemize}
\item Training sample $\mathcal{D}_{{\tiny \mbox{train}}}$: observed and used to build predictors; 
\item Calibration sample $\mathcal{D}_{{\tiny \mbox{cal}}}=\{(X_i,Y_i),i\in \range{n} \}$;

  observed and used to calibrate the size(s) of the prediction intervals;
\item Test sample $\mathcal{D}_{{\tiny \mbox{test}}}=\{(X_{n+i},Y_{n+i}), i \in \range{m}\}$; 
only the $X_i$'s are observed and the aim is to provide prediction intervals for the labels.
\end{itemize}
In addition, we consider the following 
transfer learning setting: while the data points are i.i.d. within each sample and the distributions of $\mathcal{D}_{{\tiny \mbox{cal}}}$ and $\mathcal{D}_{{\tiny \mbox{test}}}$ are the same,  the distribution of $\mathcal{D}_{{\tiny \mbox{train}}}$ can be different. However, $\mathcal{D}_{{\tiny \mbox{train}}}$ can still help to build a good predictor by using a transfer learning toolbox, considered here as a black box 
{(see, e.g., \citealp{zhuang2020comprehensive} for a survey on transfer learning)}.
{A typical situation of use is when the training labeled data $\mathcal{D}_{{\tiny \mbox{train}}}$ is abundant but
there is a domain shift for the test data, and we have a limited number of labeled data $\mathcal{D}_{{\tiny \mbox{cal}}}$ from the new domain.}

\subsection{Adaptive scores and procedures}

Formally, the aim is to build $\bm{\mathcal{I}}=(\mathcal{I}_i)_{i\in \range{m}}$, a family of $m$ random intervals of $\R$ such that the amount of coverage errors $(\ind{Y_{n+i}\notin \PI_i})_{i\in \range{m}}$ is controlled. 
The construction of a rule $\bm{\mathcal{I}}$ is based on non-conformity scores $S_i$, $1\leq i\leq n+m$,  corresponding to residuals between $Y_{i}$ and the prediction at point $X_{i}$: 
\begin{equation}\label{equadaptscores}
  S_i:=|Y_{i}-\hat{\mu}(X_{i};(\dtrain,\dxct))|, \:\: i\in \range{n+m},
\end{equation}
where the predictor $\hat{\mu}$ 
is learnt using $\dtrain$ and the calibration+test covariates $\dxct=(X_1,\ldots,X_{n+m})$.
{More sophisticated scores than the residuals have been proposed in earlier literature \citep{romano2019conformalized, gupta2022nested}, in particular allowing for
  conditional variance or quantile prediction and resulting prediction intervals of varying length. Our theory extends to those as well and we consider here~\eqref{equadaptscores} for
simplicity.}
We call the scores \eqref{equadaptscores} {\it adaptive} because they can use the unlabeled data $\dxct$, which is particularly suitable in the transfer learning framework where the covariates of $\mathcal{D}_{{\tiny \mbox{train}}}$ should be mapped to those of $\dxct$ to build a good predictor. 
Classical scores can also be recovered via \eqref{equadaptscores}  if the predictor ignores $\dxct$.
The predictor $\wh{\mu}$ can be any ``black box" (an unspecified transfer learning algorithm)
provided the following mild assumption is satisfied, ensuring score exchangeability:
\begin{multline}
\mbox{$\forall x\in \R^d$, $\hat{\mu}\big(x;(\mathcal{D}_{{\tiny \mbox{train}}}, \dxct)\big)$}
\mbox{ is invariant by permutation of } \dxct. 
\label{as:permut}\tag{PermInv}
\end{multline}
Since $(X_i,Y_i)_{i \in \range{n+m}}$ are i.i.d. and thus exchangeable, one can  easily show that  \eqref{as:exchangeable} holds for the adaptive scores \eqref{equadaptscores} when the predictor satisfies \eqref{as:permut}.
Predictors based on transfer machine learning procedures typically satisfy \eqref{as:permut}.
In addition, \eqref{as:noties} is a mild assumption: add a negligible noise to the scores is an appropriate tie breaking that makes  \eqref{as:noties} hold.

Given the scores \eqref{equadaptscores}, we build the conformal $p$-values via \eqref{equemppvalues} and define the specific conformal procedure $\bm{\mathcal{C}}(\alpha)=(\mathcal{C}_i(\alpha))_{i\in \range{m}}$ obtained by inverting $\{p_i>\alpha\}$ with respect to $Y_{n+i}$, that is, $\{p_i>\alpha\}=\{Y_{n+i}\in \mathcal{C}_i(\alpha)\}$  almost surely with 
\begin{equation}\label{equrules}
  \mathcal{C}_i(\alpha):=\left[ \hat{\mu}(X_{n+i};
  (\dtrain,\dxct)) \pm S_{(\lceil (n+1)(1-\alpha)\rceil)}\right], 
\end{equation}
where $S_{(1)}\leq \dots \leq S_{(n)}\leq S_{(n+1)}:=+\infty$ denote the order statistics of the calibration scores $(S_1,\dots,S_n)$. 
Observe that the radius of the interval $S_{(\lceil (n+1)(1-\alpha)\rceil)}$ can be equivalently described as the $(1-\alpha)$-quantile of the distribution $\sum_{i=1}^n \frac{1}{n+1}\delta_{S_i} + \frac{1}{n+1} \delta_{+\infty}$. Note also that $\bm{\mathcal{C}}(\alpha)=\R^m$ if $\alpha<1/(n+1)$, that is, if the desired coverage error is too small w.r.t. the size of the calibration sample.

\subsection{Transductive error rates}

By Proposition~\ref{prop:exch}, the following marginal control holds for the conformal procedure $\bm{\mathcal{C}}(\alpha)$ \eqref{equrules}:
\begin{equation}\label{equ:marg}
\P(Y_{n+i}\notin \mathcal{C}_i(\alpha))\leq \alpha, \:\:i\in \range{m}.
\end{equation}
This is classical for non-adaptive scores and our result already brings an extension to adaptive scores in the transfer learning setting. 

In addition, we take into account the prediction multiplicity by considering {\it false coverage proportion} (FCP) of some procedure $\bm{\mathcal{I}}=(\PI_i)_{i\in \range{m}}$, given by
\begin{align}
\FCP(\bm{\mathcal{I}})&:=m^{-1}\sum_{i=1}^m \ind{Y_{n+i}\notin \PI_i}.\label{error}
\end{align}
It is clear from~\eqref{equ:marg} that the procedure $\bm{\mathcal{C}}(\alpha)$ \eqref{equrules} controls the {\it false coverage rate}, that is, $\FCR(\bm{\mathcal{C}}(\alpha))):=\E[\FCP(\bm{\mathcal{C}}(\alpha))]\leq \alpha$. 
However,  the error $\FCP(\bm{\mathcal{C}}(\alpha))$ naturally fluctuates around its mean and the event $\{\FCP(\bm{\mathcal{C}}(\alpha))\leq \alpha\}$ is not guaranteed. Hence, we aim at the following control in probability of the FCP:
\begin{align}
\P[\FCP(\bm{\mathcal{C}}(\alpha))\leq \ol{\alpha}]&\geq 1-\delta\label{controlalphaFWER}.
\end{align}
Several scenarios can be considered: $\alpha$ is fixed and we want to find a suitable bound $\ol{\alpha} = \ol{\FCP}_{\alpha,\delta}$ for the
``traditional'' conformal procedure $\bm{\mathcal{C}}(\alpha)$; or conversely, $\ol{\alpha}$ is fixed and we want to adjust the
parameter $\alpha = t_{\ol{\alpha},\delta}$ of the procedure to ensure the probabilistic control at target level $\ol{\alpha}$.
For $\ol{\alpha}=0$, this reduces to $\P[\forall i\in\range{m}, \:Y_{n+i}\in \PI_i]\geq 1-\delta$, i.e., no false coverage with high probability. By applying a union bound, the procedure $\bm{\mathcal{C}}(\delta/m)$ 
satisfies the latter control, as already proposed by \cite{vovk2013transductive}. However,  in this case the predicted intervals can be trivial, that is,  $\bm{\mathcal{C}}(\delta/m)=\R^m$, if the test sample is too large, namely, $m> \delta(n+1)$.
Moreover, in a more general scenario the practitioner may want to adjust the parameter $\alpha=\wh{\alpha}$ on their own
depending on the data, for example based on some personal tradeoff between the probabilistic control obtained and the length of the corresponding prediction intervals ---
this is the common practice of a ``post-hoc'' choice (made after looking at the data).
This motivates us to aim at a uniform (in $\alpha$) bound, that is,
find a family of random variables $({\overline{\FCP}}_{\alpha,\delta})_{\alpha\in (0,1)}$ such that 
\begin{align}
\prob{\forall \alpha\in (0,1),\:\: \FCP(\bm{\mathcal{C}}(\alpha))\leq {\overline{\FCP}}_{\alpha,\delta}}&\geq 1-\delta\:.\label{controlalphaFWERunif}
\end{align}
Establishing such bounds is investigated in the next section.
This gives a guarantee on the FCP in any of the above scenarios, in particular a post-hoc choice of the parameter $\wh{\alpha}$.
\begin{figure}[t]
\vspace{-1cm}
\begin{center}
\hspace{-6mm}\includegraphics[scale=0.38]{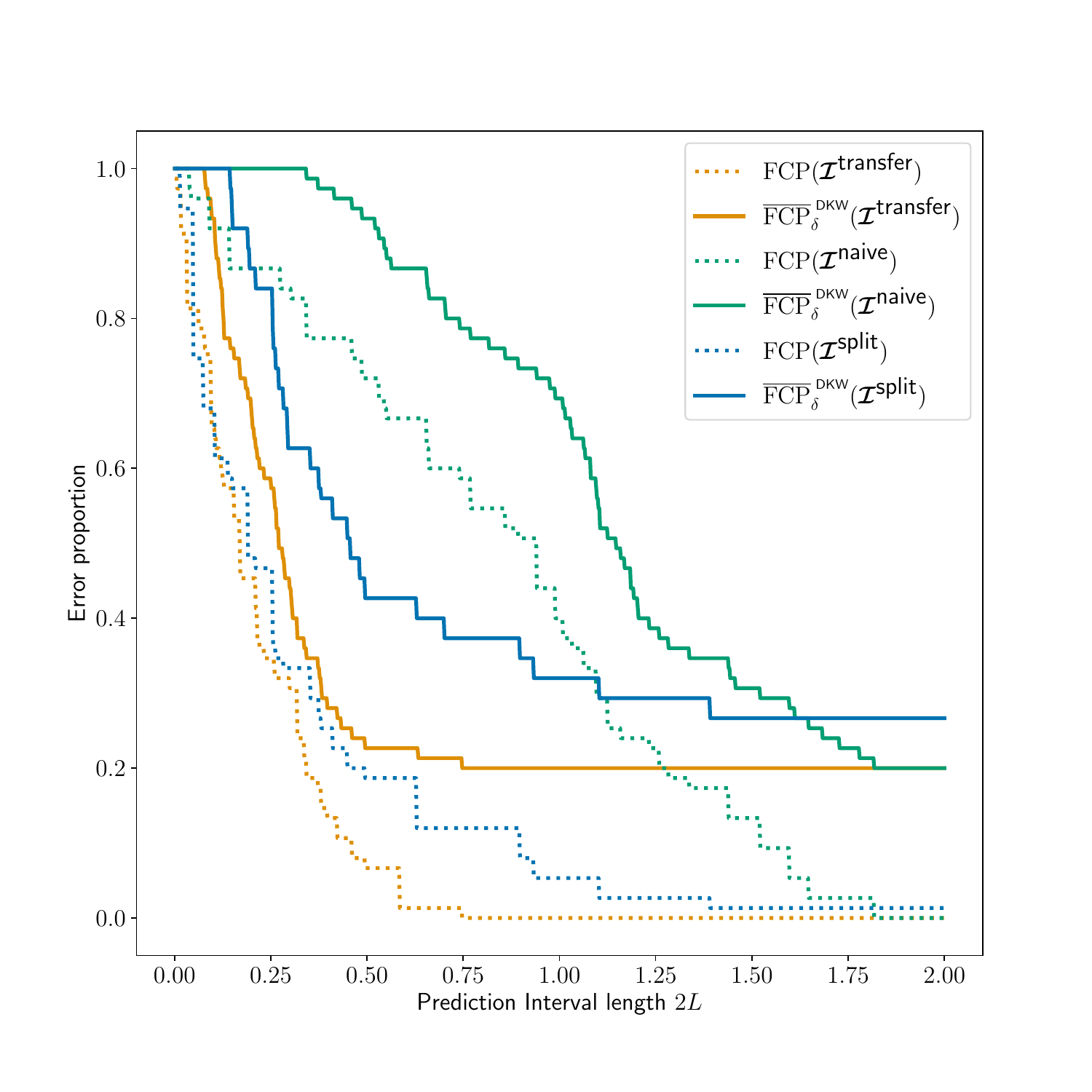}
\end{center}
\vspace{-1cm}
\caption{
Plot of $\FCP(\bm{\mathcal{I}})$ \eqref{error} (dashed) and bound ${\overline{\FCP}}^{\mbox{\tiny DKW}}_{\hat{\alpha}(L),\delta}$  \eqref{boundfalsepositiveDKW} \eqref{equ:length} (solid, $\delta=0.2$) in function of interval length $2L$ in  the same setting and procedures as in Figure~\ref{fig:IlluTransfert}. 
\label{fig:IlluAlpha_L}} 
\end{figure}
As a concrete example, one may want to choose a data-dependent $\wh{\alpha}$ to ensure prediction intervals $\bm{\mathcal{C}}(\alpha)$ of radius at most $L$, namely, 
{\begin{equation}\label{equ:length}
\widehat{\alpha}(L)=(n+1)^{-1}\Big(1+\sum_{i=1}^n \ind{S_i> L}\Big).
\end{equation}
}
Guarantee \eqref{controlalphaFWERunif} yields a $(1-\delta)$-confidence error bound  ${\overline{\FCP}}_{\widehat{\alpha}(L),\delta}$ for this choice.

\subsection{Controlling the error rates}\label{sec:controlPI}

To establish \eqref{controlalphaFWER} and \eqref{controlalphaFWERunif}, we use that from \eqref{equ-ecdfpvalues}, \eqref{equrules} and \eqref{error}, $\FCP(\bm{\mathcal{C}}(t))=\wh{F}_m(t)$ and thus for all $t\in [0,1]$,
\begin{align*}
\{\FCP(\bm{\mathcal{C}}(t))\leq \ol{\alpha}\}&=\set[1]{\wh{F}_m(t) \leq \ol{\alpha}}\\
&=\set[1]{m\wh{F}_m(t) \leq \lfloor \ol{\alpha} m\rfloor}\\
&=\set[1]{p_{(\lfloor \ol{\alpha} m\rfloor+1)}>t},
\end{align*}
where $p_{(1)}\leq \dots\leq p_{(m)}$ denote an ordered conformal $p$-values.
We deduce the following result. 

\begin{corollary}\label{corsimultaneous}
Let $n,m\geq 1$. Consider the setting of Section~\ref{sec:settingconformal}, the conformal procedure $\bm{\mathcal{C}}(\alpha)$ given by \eqref{equrules} and $P_{n,m}$ given by \eqref{equ:distribution}. Then the following holds:
\begin{itemize}
\item[(i)] for any $\ol{\alpha}\in [0,1]$, $\delta\in (0,1)$,  $\bm{\mathcal{C}}(\alpha=t_{\ol{\alpha},\delta})$ satisfies \eqref{controlalphaFWER} provided that $t_{\ol{\alpha},\delta}$ is chosen s.t.
\begin{equation}\label{equalphabaredelta}
\P_{\bm{p}\sim P_{n,m}}( p_{(\lfloor \ol{\alpha} m\rfloor+1)}\leq t_{\ol{\alpha},\delta})\leq \delta.
\end{equation}
\item[(ii)] for any $\delta\in (0,1)$,  $\paren[1]{{\overline{\FCP}}_{\alpha,\delta}}_{\alpha\in (0,1)}$ satisfies  \eqref{controlalphaFWERunif}  provided that 
\begin{equation}\label{equalphabaredeltabis}
\P_{\bm{p}\sim P_{n,m}}\paren[1]{\exists \alpha\in (0,1)\::\: \wh{F}_m(\alpha) > {\overline{\FCP}}_{\alpha,\delta} }\leq \delta.
\end{equation}
\end{itemize}
\end{corollary}

Applying Corollary~\ref{corsimultaneous} (i), for conformal prediction with guaranteed FCP, we obtain an adjusted level parameter 
which can be computed numerically
(an explicit formula can also be given for $\alpha=0$, see Supplemental~\ref{sec:minp}). Applying Corollary~\ref{corsimultaneous} (ii), and thanks to \eqref{boundDKWup}, the following family bound $({\overline{\FCP}}_{\alpha,\delta})_{\alpha\in (0,1)}$  is valid for \eqref{controlalphaFWERunif}
\begin{align}
{\overline{\FCP}}^{\mbox{\tiny DKW}}_{\alpha,\delta}&= \paren[1]{\alpha + \lambda^{\mbox{\tiny DKW}}_{\delta,n,m}}\ind{\alpha\geq 1/(n+1)},  \label{boundfalsepositiveDKW}
\end{align}
with $\lambda^{\mbox{\tiny DKW}}_{\delta,n,m}>0$ given by \eqref{boundDKWupexplicit}.
Obviously, numerical bounds can also be developed according to Remark~\ref{numispossible}.

\begin{remark}\label{rem:selection}
{Our FDP bounds extend to a selective inference framework where $\mathcal{S}=\mathcal{S}(\mathcal{D}_{{\tiny \mbox{train}}}, \dxct))\subset \range{n+m}$ is a selection rule invariant by permutation of $\dxct$, typically $\mathcal{S}=\{i\in \range{n+m}\::\:  \hat{\mu}(X_{i})\geq 0\}$. The calibration and test samples {\it over the selection $S$} are $\mathcal{D}_{\mathcal{S},{\tiny \mbox{cal}}}=\{(X_i,Y_i),i\in \mathcal{S}\cap \range{n} \}$ and $\mathcal{D}_{\mathcal{S},{\tiny \mbox{test}}}=\{(X_{n+i},Y_{n+i}), i \in \mathcal{S}\cap \range{n+1,n+m}\}$, respectively. Defining the $p$-values accordingly, our envelopes are also valid for the FCP {\it over the selection $S$} by simply replacing 
 $n$ by $|\mathcal{S}\cap \range{n}|$ and $m$ by $|\mathcal{S}\cap \range{n+1,n+m}|$. 
This complements the recent work of \cite{bao2023selective}, where only in-expectation results were established.}
\end{remark}

\subsection{Numerical experiments}\label{sec:numexpPI}

To illustrate the performance of the method, we consider the following proof-of-concept regression model: $(W_i,Y_i)$ i.i.d. with $Y_i \:|\: W_i \sim \mathcal{N}(\mu(W_i),\sigma^2)$  for some unknown function $\mu$ and parameter $\sigma>0$.
To accommodate the transfer learning setting, we assume that we observe $X_i=f_1(W_i)$ in $\dtrain$ and $X_i=f_2(W_i)$ in $\dcal \cup \dtest$ for some transformations $f_1$ and $f_2$. 
Three conformal procedures\footnote{Python code for (PI) based on implementation of \cite{BoyerCP}.} $\bm{\mathcal{I}}=\bm{\mathcal{C}}(\alpha)=(\mathcal{C}_i(\alpha))_{i\in \range{m}}$ are considered which differ only in the construction of the scores: first, $\bm{\mathcal{I}}^{\tiny \mbox{naive}}$ consists in using a predictor of the usual form $\hat{\mu}(\cdot,\mathcal{D}_{{\tiny \mbox{train}}})$ hence ignoring the distribution difference between $\dtrain$ and $\dcal \cup \dtest$ (no transfer) with a RBF kernel ridge regression;
the second procedure $\bm{\mathcal{I}}^{\tiny \mbox{split}}$ ignores completely $\dtrain$ and works by splitting $\dcal$ in two new samples of equal size to apply the usual approach with these new (reduced) samples (transfer not needed); the third approach $\bm{\mathcal{I}}^{\tiny \mbox{transfer}}$ is the proposed one, and uses the transfer predictor $\hat{\mu}(\cdot;(\dtrain,\dxct))$
 based on optimal transport
proposed by \cite{courty2017joint}.
While all methods provide the correct $(1-\alpha)$ marginal coverage, we see from Figure~\ref{fig:IlluTransfert} that $\bm{\mathcal{I}}^{\tiny \mbox{transfer}}$ is much more accurate, which shows the benefit of using transfer learning and adaptive scores. Here, $\left |\dtrain\right |=5000$, 
$n=m=75 $, $\mu(x)=\cos(x)$,  $W_i\sim\mathcal{U}(0,5)$, $f_1(x)=x$, $f_2(x)=0.6x+x^{2}/25$ and $\sigma=0.1$.
Next, for each of the three methods, the FCP and corresponding bounds \eqref{boundfalsepositiveDKW} are displayed in Figure~\ref{fig:IlluAlpha_L}. This illustrates both that each bound is uniformly valid in $L$ and that transfer learning reduces the FCP (and thus also the FCP bounds).

\section{Application to novelty detection}\label{sec:ND}

\subsection{Setting}\label{sec:settingnd}

In the novelty detection problem, we 
observe the two following independent samples:
\begin{itemize}
\item a training null sample $\mathcal{D}_{{\tiny \mbox{null}}}$ of $n_0$ nominal data points in $\R^d$ which are i.i.d. with common distribution $P_0$;
\item a test sample $\mathcal{D}_{{\tiny \mbox{test}}}=(X_i, i\in \range{m})$ of independent points in $\R^d$ either distributed as $P_0$ or not.
\end{itemize}

The aim is to decide if each $X_i$ is distributed as the training sample (that is, as $P_0$) or not. 
This long standing problem in machine learning has been recently revisited with the aim of controlling the proportion of errors among the items declared as novelties \cite{bates2023testing}; let $\cH_0=\{i\in \range{m}\::\: X_i\sim P_0\}$ corresponding to the set of non-novelty in the test sample and consider the false discovery proportion
\begin{equation}\label{FDP}
\FDP(R)=\frac{|R\cap \cH_0|}{|R|\vee 1},
\end{equation}
for any (possibly random) subset $R\subset \range{m}$ corresponding to the $X_i$'s declared as novelties. 
The advantage of considering $\FDP(R)$ for measuring the errors has been widely recognized in the multiple testing literature since the fundamental work of \cite{BH1995} and its popularity is nowadays increasing in large scale machine learning theory, see \cite{bates2023testing,marandon2022machine,jin2023model,bashari2023derandomized}, among others. 
The main advantage of $\FDP(R)$ is that the number of errors  $|R\cap \cH_0|$ is rescaled by the number of declared novelties $|R|$, which makes it scale invariant with respect to the size $m$ of the test sample, so that novelty detection can still be possible in large scale setting.

\subsection{Adaptive scores}

Following \cite{bates2023testing,marandon2022machine}, we assume that scores are computed as follows:
\begin{enumerate} 
\item Split the null sample $\mathcal{D}_{{\tiny \mbox{null}}}$ into $\mathcal{D}_{{\tiny \mbox{train}}}$ and  $\mathcal{D}_{{\tiny \mbox{cal}}}=(X_i, i\in \range{n})$ for some chosen $n\in (1,n_0)$;
\item Compute novelty scores $S_{i}=g(X_{i})$, $i\in \range{n+m}$, 
for some score function $g:\R^d\to \R$ (discussed below);
\item Compute conformal $p$-values as in \eqref{equemppvalues}. 
\end{enumerate}
In the work of \cite{bates2023testing}, the score function is built from $\mathcal{D}_{{\tiny \mbox{train}}}$ only, using a one-class classification method (classifier solely based on null examples), which makes the scores independent conditional to $\dtrain$. The follow-up work \cite{marandon2022machine} considers a score function depending both on $\mathcal{D}_{{\tiny \mbox{train}}}$ and $\mathcal{D}_{{\tiny \mbox{cal}}}\cup \mathcal{D}_{{\tiny \mbox{test}}}$ (in a permutation-invariant way of the sample $\mathcal{D}_{{\tiny \mbox{cal}}}\cup \mathcal{D}_{{\tiny \mbox{test}}}$), which allows to use a two-class classification method including test examples. Doing so, the scores are adaptive to the form of the novelties present in the test sample, which 
significantly improves novelty detection (in a nutshell: it is much easier to detect an object when we have some examples of it). 
While the independence of the scores is lost, an appropriate exchangeability property is maintained so that we can apply our theory in that case, by assuming in addition \eqref{as:noties}.

\subsection{Methods and FDP bounds}

Let us consider any thresholding novelty procedure
\begin{equation}\label{thresrule}
\mathcal{R}(t):=\{i\in \range{m}\::\:p_i\leq t\}, \:\: t\in (0,1).
\end{equation}
Then the following result holds true.
\begin{corollary}\label{cor:ThresholdFDP}
In the above novelty detection setting and under Assumption~\ref{as:noties}, the family of thresholding novelty procedures \eqref{thresrule} is such that, with probability at least $1-\delta$, we have for all  $ t\in (0,1)$,
\begin{equation}\label{ThresholdFDPboundDKW_nonStorey}
\FDP(\mathcal{R}(t))\leq  \frac{{m} I_n(t) + {m} \lambda^{\mbox{\tiny DKW}}_{\delta,n,{m}}}{1\vee |\mathcal{R}(t)|}=:{\overline{\FDP}}^{\tiny \mbox{DKW}}_{t,\delta},
\end{equation}
and with an estimation of $m_0$,
\begin{align}
&\FDP(\mathcal{R}(t)) \label{ThresholdFDPboundDKW}\leq  \frac{\hat{m}_0 I_n(t) +\max_{r\in\range{\hat{m}_0}} \{r \lambda^{\mbox{\tiny DKW}}_{\delta,n,r}\}}{1\vee |\mathcal{R}(t)|}=:{\overline{\FDP}}^{\tiny \mbox{DKW}}_{t,\delta},
\end{align}
where $\lambda^{\mbox{\tiny DKW}}_{\delta,n,r}$ is given by \eqref{boundDKWupexplicit} and
$\hat{m}_0$ is any random variable such that
\begin{align}
&\hat{m}_0\geq  \label{m0hat}
\max\bigg\{r:\inf_t\frac{\sum_{i=1}^m\ind{p_i> t} + \max_{u\in \range{r}}\{u \lambda^{\mbox{\tiny DKW}}_{\delta,n,u}\} }{1-I_{n}(t)}\geq r\bigg\},
\end{align}
where $r$ is in the range $\range{m}$ and the maximum is equal to $m$ if the set is empty. 
\end{corollary}
The proof is provided in Supplemental~\ref{sec:adadetect}.

\begin{remark}\label{rem:FDPboundalpha}
Among thresholding procedures \eqref{thresrule}, AdaDetect \citep{marandon2022machine} is obtained by applying the Benjamini-Hochberg (BH) procedure \citep{BH1995} to the conformal $p$-values. It is proved to control the expectation of the FDP (that is, the false discovery rate, FDR) at level $\alpha$.
Applying Corollary~\ref{cor:ThresholdFDP} provides in addition an FDP bound for AdaDetect, uniform in $\alpha$, see  Supplemental~\ref{sec:AdaDetect}.
\end{remark}

\begin{figure}[h!]
\vspace{-1cm} 
\begin{center}
\hspace{-6mm}\includegraphics[scale=0.38]{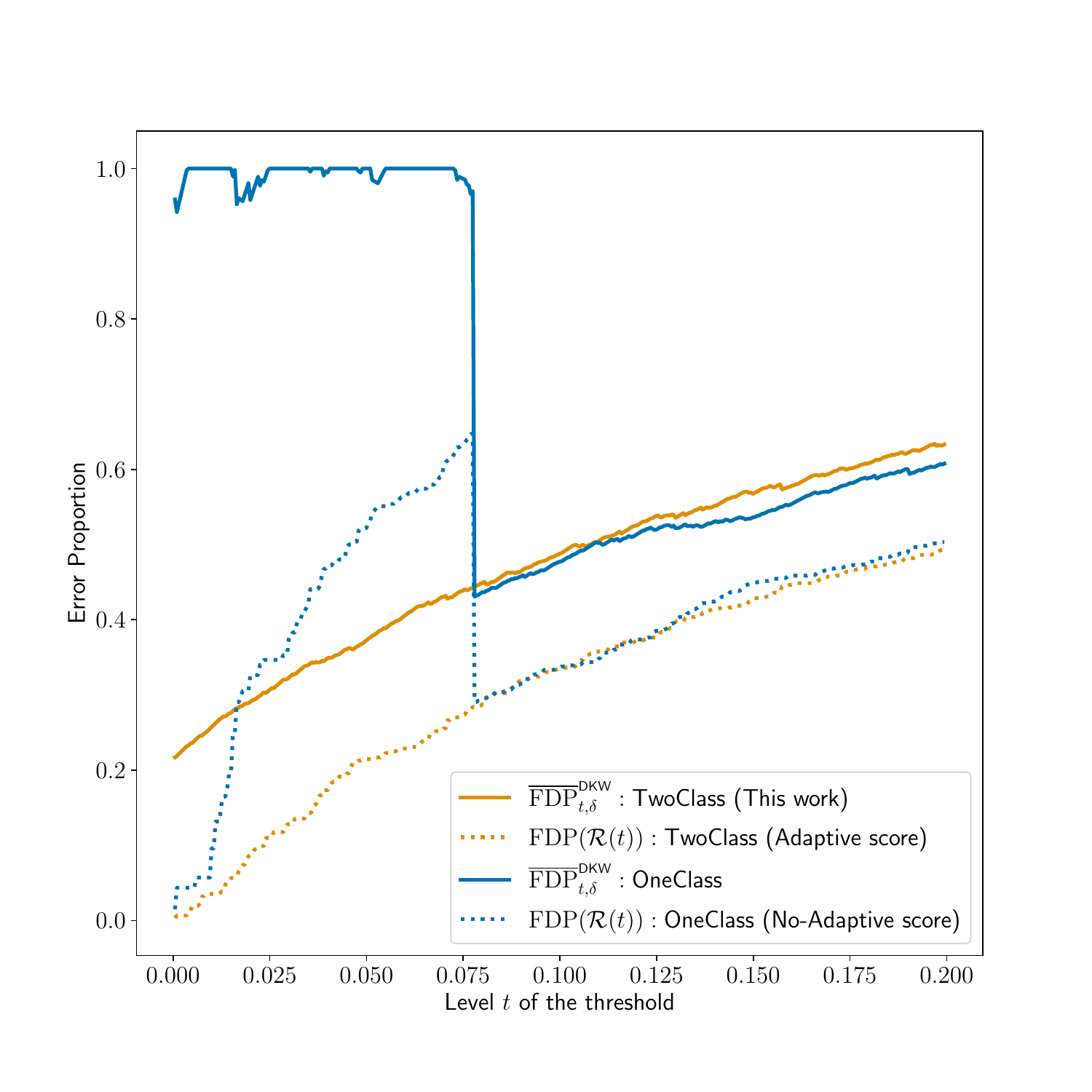}
\end{center}
\vspace{-1cm}
\caption{
  Plot of $\FDP(\mathcal{R}(t))$ \eqref{FDP}\eqref{thresrule} (dashed) and bound ${\overline{\FDP}}^{\mbox{\tiny DKW}}_{t,\delta}$  \eqref{ThresholdFDPboundDKW}  (solid, $\delta=0.2$) in function of the threshold $t$ for $\mathcal{R}(t)$ \eqref{thresrule} with a score obtained either with a one-class classification (non-adaptive) or a two-class classification (adaptive).
\label{fig:nd}}
\end{figure}

\subsection{Numerical experiments}\label{sec:numexpnd}

We follow the numerical experiments on ``Shuttle'' datasets of \cite{marandon2022machine}\footnote{The Python code uses the implementation of the procedure AdaDetect of  \cite{marandonAdaImplementation}.}.  
In Figure~\ref{fig:nd}, we displayed the true FDP and the corresponding bound \eqref{ThresholdFDPboundDKW} when computing $p$-values based on different scores: the non-adaptive scores of \cite{bates2023testing} obtained with  isolation forest one-class classifier; and the adaptive scores of \cite{marandon2022machine} obtained with random forest two-class classifier. While the advantage of considering adaptive scores is clear (smaller FDP and bound)
, it illustrates that the bound is correct simultaneously on $t$. 
Additional experiments are provided in Supplemental~\ref{sec:addexp}.

\section{Conclusion}

{The main takeaway from this work is the characterization of a ``universal'' joint distribution $P_{n,m}$ for conformal $p$-values
based on $n$ calibration points and  $m$ test points.
We derived as a consequence a non-asymptotic
concentration inequality for the $p$-value empirical distribution function; numerical procedures can also be of use for calibration in practice.
This entails uniform error bounds on the false coverage/false discovery proportion that hold with high probability, while standard results are only marginal or in expectation and not uniform in the decision. Since the results hold under the
score exchangeability assumption only, they are applicable to {\it adaptive} score procedures using the calibration and test sets for training.

\section*{Acknowledgements}

We would like to thank Anna Benhamou for constructive discussions and Romain P\'erier for  pointing out an error in a previous version of the manuscript.
The authors acknowledge the grants ANR-21-CE23-0035 (ASCAI), ANR-23-CE40-0018-01 (BACKUP)
and ANR-23-CE40-0018-01 (BISCOTTE) of the French National Research Agency ANR, and the Emergence project MARS of Sorbonne Universit\'e.

\bibliography{biblio}

\begin{thebibliography}{}

\bibitem[Balasubramanian et~al., 2014]{balasubramanian2014conformal}
Balasubramanian, V., Ho, S.-S., and Vovk, V. (2014).
\newblock {\em Conformal prediction for reliable machine learning: theory,
  adaptations and applications}.
\newblock Morgan Kaufmann books.

\bibitem[Bao et~al., 2024]{bao2023selective}
Bao, Y., Huo, Y., Ren, H., and Zou, C. (2024).
\newblock {Selective conformal inference with false coverage-statement rate
  control}.
\newblock {\em Biometrika}, page asae010.

\bibitem[Bashari et~al., 2023]{bashari2023derandomized}
Bashari, M., Epstein, A., Romano, Y., and Sesia, M. (2023).
\newblock Derandomized novelty detection with {FDR} control via conformal
  {E}-values.
\newblock {\em arXiv preprint 2302.07294}.

\bibitem[Bates et~al., 2023]{bates2023testing}
Bates, S., Cand{\`e}s, E., Lei, L., Romano, Y., and Sesia, M. (2023).
\newblock Testing for outliers with conformal p-values.
\newblock {\em Ann. Statist.}, 51(1):149--178.

\bibitem[Benjamini and Hochberg, 1995]{BH1995}
Benjamini, Y. and Hochberg, Y. (1995).
\newblock Controlling the false discovery rate: a practical and powerful
  approach to multiple testing.
\newblock {\em J. Roy. Statist. Soc. Ser. B}, 57(1):289--300.

\bibitem[Benjamini and Yekutieli, 2001]{BY2001}
Benjamini, Y. and Yekutieli, D. (2001).
\newblock The control of the false discovery rate in multiple testing under
  dependency.
\newblock {\em Ann. Statist.}, 29(4):1165--1188.

\bibitem[Bian and Barber, 2022]{bian2022training}
Bian, M. and Barber, R.~F. (2022).
\newblock Training-conditional coverage for distribution-free predictive
  inference.
\newblock {\em arXiv preprint arXiv:2205.03647}.

\bibitem[Bian and Barber, 2023]{bian2023training}
Bian, M. and Barber, R.~F. (2023).
\newblock Training-conditional coverage for distribution-free predictive
  inference.
\newblock {\em Electronic Journal of Statistics}, 17(2):2044--2066.

\bibitem[{Blanchard} et~al., 2020]{BNR2020}
{Blanchard}, G., {Neuvial}, P., and {Roquain}, E. (2020).
\newblock {Post hoc confidence bounds on false positives using reference
  families}.
\newblock {\em {Ann. Statist.}}, 48(3):1281--1303.

\bibitem[Boyer and Zaffran, 2023]{BoyerCP}
Boyer, C. and Zaffran, M. (2023).
\newblock Tutorial on conformal prediction.
\newblock \url{https://claireboyer.github.io/tutorial-conformal-prediction/}.

\bibitem[Courty et~al., 2017]{courty2017joint}
Courty, N., Flamary, R., Habrard, A., and Rakotomamonjy, A. (2017).
\newblock Joint distribution optimal transportation for domain adaptation.
\newblock In {\em Advances in neural information processing systems 30 (NIPS
  2017)}, volume~30.

\bibitem[Dal~Pozzolo et~al., 2015]{dal2015calibrating}
Dal~Pozzolo, A., Caelen, O., Johnson, R.~A., and Bontempi, G. (2015).
\newblock Calibrating probability with undersampling for unbalanced
  classification.
\newblock In {\em 2015 IEEE symposium series on computational intelligence},
  pages 159--166. IEEE.

\bibitem[Gupta et~al., 2022]{gupta2022nested}
Gupta, C., Kuchibhotla, A.~K., and Ramdas, A. (2022).
\newblock Nested conformal prediction and quantile out-of-bag ensemble methods.
\newblock {\em Pattern Recognition}, 127:108496.

\bibitem[Huang et~al., 2023]{huang2023uncertainty}
Huang, K., Jin, Y., Candes, E., and Leskovec, J. (2023).
\newblock Uncertainty quantification over graph with conformalized graph neural
  networks.
\newblock {\em Advances in Neural Information Processing Systems}, 36.

\bibitem[Jin and Cand{\`e}s, 2023]{jin2023model}
Jin, Y. and Cand{\`e}s, E.~J. (2023).
\newblock Model-free selective inference under covariate shift via weighted
  conformal p-values.
\newblock {\em arXiv preprint arXiv:2307.09291}.

\bibitem[Lei et~al., 2018]{lei2018distribution}
Lei, J., G'Sell, M., Rinaldo, A., Tibshirani, R.~J., and Wasserman, L. (2018).
\newblock Distribution-free predictive inference for regression.
\newblock {\em J. Amer. Stat. Assoc.}, 113(523):1094--1111.

\bibitem[Li et~al., 2022]{li2022simultaneous}
Li, J., Maathuis, M.~H., and Goeman, J.~J. (2022).
\newblock Simultaneous false discovery proportion bounds via knockoffs and
  closed testing.
\newblock {\em arXiv preprint arXiv:2212.12822}.

\bibitem[Marandon, 2022]{marandonAdaImplementation}
Marandon, A. (2022).
\newblock Machine learning meets {FDR}.
\newblock
  \url{https://github.com/arianemarandon/adadetect#machine-learning-meets-fdr}.

\bibitem[Marandon et~al., 2022]{marandon2022machine}
Marandon, A., Lei, L., Mary, D., and Roquain, E. (2022).
\newblock Machine learning meets false discovery rate.
\newblock {\em arXiv preprint 2208.06685}.

\bibitem[{Marques F.}, 2023]{f2023universal}
{Marques F.}, P.~C. (2023).
\newblock On the universal distribution of the coverage in split conformal
  prediction.
\newblock {\em arXiv preprint 2303.02770}.

\bibitem[Massart, 1990]{Mass1990}
Massart, P. (1990).
\newblock The tight constant in the {D}voretzky-{K}iefer-{W}olfowitz
  inequality.
\newblock {\em Ann. Probab.}, 18(3):1269--1283.

\bibitem[Papadopoulos et~al., 2002]{papadopoulos2002inductive}
Papadopoulos, H., Proedrou, K., Vovk, V., and Gammerman, A. (2002).
\newblock Inductive confidence machines for regression.
\newblock In {\em 13th European Conference on Machine Learning (ECML 2002)},
  pages 345--356. Springer.

\bibitem[Romano and Wolf, 2005]{RW2005}
Romano, J.~P. and Wolf, M. (2005).
\newblock Exact and approximate stepdown methods for multiple hypothesis
  testing.
\newblock {\em J. Amer. Statist. Assoc.}, 100(469):94--108.

\bibitem[Romano et~al., 2019]{romano2019conformalized}
Romano, Y., Patterson, E., and Candes, E. (2019).
\newblock Conformalized quantile regression.
\newblock {\em Advances in neural information processing systems}, 32.

\bibitem[Sarkar and Kuchibhotla, 2023]{sarkar2023post}
Sarkar, S. and Kuchibhotla, A.~K. (2023).
\newblock Post-selection inference for conformal prediction: Trading off
  coverage for precision.
\newblock {\em arXiv preprint arXiv:2304.06158}.

\bibitem[Saunders et~al., 1999]{saunders1999transduction}
Saunders, C., Gammerman, A., and Vovk, V. (1999).
\newblock Transduction with confidence and credibility.
\newblock In {\em 16th International Joint Conference on Artificial
  Intelligence (IJCAI 1999)}, pages 722--726.

\bibitem[Simes, 1986]{Sim1986}
Simes, R.~J. (1986).
\newblock An improved {B}onferroni procedure for multiple tests of
  significance.
\newblock {\em Biometrika}, 73(3):751--754.

\bibitem[Vanschoren et~al., 2013]{OpenML2013}
Vanschoren, J., van Rijn, J.~N., Bischl, B., and Torgo, L. (2013).
\newblock Openml: networked science in machine learning.
\newblock {\em SIGKDD Explorations}, 15(2):49--60.

\bibitem[Vovk, 2012]{vovk2012conditional}
Vovk, V. (2012).
\newblock Conditional validity of inductive conformal predictors.
\newblock In {\em 4th Asian conference on machine learning (ACML 2012)}, pages
  475--490. PMLR.

\bibitem[Vovk, 2013]{vovk2013transductive}
Vovk, V. (2013).
\newblock Transductive conformal predictors.
\newblock In {\em Artificial Intelligence Applications and Innovations: 9th
  IFIP WG 12.5 International Conference (AIAI 2013)}, pages 348--360. Springer.

\bibitem[Vovk et~al., 2005]{vovk2005algorithmic}
Vovk, V., Gammerman, A., and Shafer, G. (2005).
\newblock {\em Algorithmic learning in a random world}.
\newblock Springer.

\bibitem[Woods et~al., 1993]{woods1993comparative}
Woods, K.~S., Doss, C.~C., Bowyer, K.~W., Solka, J.~L., Priebe, C.~E., and
  Kegelmeyer~Jr, W.~P. (1993).
\newblock Comparative evaluation of pattern recognition techniques for
  detection of microcalcifications in mammography.
\newblock {\em International Journal of Pattern Recognition and Artificial
  Intelligence}, 7(06):1417--1436.

\bibitem[Zhuang et~al., 2020]{zhuang2020comprehensive}
Zhuang, F., Qi, Z., Duan, K., Xi, D., Zhu, Y., Zhu, H., Xiong, H., and He, Q.
  (2020).
\newblock A comprehensive survey on transfer learning.
\newblock {\em Proceedings of the IEEE}, 109(1):43--76.

\end{thebibliography}
\vfill

\pagebreak

\appendix

\section{Exact formulas for $P_{n,m}$}\label{sec:polya}

In this section, we provide new formulas for the distribution $P_{n,m}$ given by \eqref{equ:distribution}.
First let for $\bm{j}=(j_1,\dots,j_m)\in \range{n+1}^m$,   $\bm{M}(\bm{j}):=(M_k(\bm{j}))_{k\in \range{n+1}}$ where  $M_k(\bm{j}):=|\{i\in \range{m}\::\: j_i=k\}|$ is the number of coordinates of $\bm{j}$ equal to $k$, for $k\in \range{n+1}$, and $\bm{M}(\bm{j})!:=\prod_{k=1}^{n+1} (M_k(\bm{j})!)$.

\begin{theorem}\label{th:key}
$P_{n,m}$ corresponds to the distribution of the colors of $m$ successive draws in a standard P\'olya urn model with $n+1$ colors labeled as $\set[1]{\frac{\l}{n+1},\l\in \range{n+1}}$ (with an urn starting with $1$ ball of each color). That is, 
for $\bm{p}\sim P_{n,m}$ in \eqref{equ:distribution}, we have
\begin{itemize}
\item[(i)] Sequential distribution: for all $i\in \range{0,m-1}$, 
the distribution of $p_{i+1}$ conditionally on $p_1,\dots,p_{i}$ does not depend on $m$ and is given by 
\begin{equation}
\mathcal{D}(p_{i+1}\:|\: p_1,\dots,p_{i}) = \sum_{j=1}^{n+1} \frac{1+\sum_{k=1}^{i}\ind{p_k=j/(n+1)}}{n+1+i}\delta_{j/(n+1)}\label{equMarkov}.
\end{equation}
\item[(ii)] Joint distribution: for all vectors $\bm{j}\in \range{n+1}^m$,
\begin{equation}
\P\paren{
\bm{p}=\frac{\bm{j}}{n+1}
} = 
\bm{M}(\bm{j})!
 \frac{n!}{(n+m)!}\label{equjoint},
\end{equation}
\item[(iii)] Histogram distribution: the histogram of $\bm{p}$ is uniformly distributed on the set of histograms of $m$-sample into $n+1$ bins, that is, for all $\bm{m}=(m_1,\dots,m_{n+1})\in \range{0,m}^{n+1}$ with $ m_1+\dots+m_{n+1}=m$, 
\begin{equation}
\P\big(\bm{M}\paren[1]{(n+1)\bm{p}}=\bm{m}\big) = \binom{n+m}{m}^{-1}.
\label{equunifhisto}
\end{equation}
In particular, conditionally on $\bm{M}\paren[1]{(n+1)\bm{p}}$, the variable $\bm{p}$ is uniformly distributed on the set of possible trajectories, that is,
 for all vectors $\bm{j}\in \range{n+1}^m$, 
\begin{equation}
\P\paren{\bm{p}=\frac{\bm{j}}{n+1} \:\bigg|\: \bm{M}\paren[1]{(n+1)\bm{p}}=\bm{M}\paren[1]{\bm{j}}} = \frac{\bm{M}\paren[1]{\bm{j}}!}{m!}.
\label{equuniftraj}
\end{equation}
\end{itemize}
\end{theorem}

Theorem~\ref{th:key} is proved in Section~\ref{proofth:key} for completeness. 
Theorem~\ref{th:key} (i) gives the mechanism of the P\'olya urn model: Namely, the urn first contains one ball of each of the $n+1$ colors, so $p_1$ has a uniform distributed on $\set[1]{\frac{\l}{n+1},\l\in \range{n+1}}$; then, given $p_1=\ell/(n+1)$, we have drawn a ball of color $\ell$ and we put back this ball in the urn with another one of the same color $\ell$, so $p_2$ is generated according to the distribution on $\set[1]{\frac{\l}{n+1},\l\in \range{n+1}}$ with equal chance ($=1/(n+2)$) of generating $k/(n+1)$, $k\neq \l$, and twice more chance ($=2/(n+2)$) of generating $\l/(n+1)$. Recursively, given $p_1,\dots,p_{i}$, the random variable $p_{i+1}$ is generated in $\set[1]{\frac{\l}{n+1},\l\in \range{n+1}}$ according to the sizes of the histogram of the sample $((n+1)p_1,\dots, (n+1)p_{i})$, see Figure~\ref{fig:dyn}.

Theorem~\ref{th:key} (ii) provides the exact dependency structure between the $p$-values: for instance, $\bm{M}(\bm{j})!=1$ when the coordinates of $\bm{j}=(j_1,\dots,j_m)$ are all distinct, while $\bm{M}(\bm{j})!=m!$ when the coordinates of $\bm{j}=(j_1,\dots,j_m)$ are the same. This means that the distribution slightly favors the $\bm{j}$ with repeated entries. This shows that the conformal $p$-values are not i.i.d. but have a positive structure of dependency. This is in accordance with the specific positive dependence property (called PRDS) already shown by \cite{bates2023testing,marandon2022machine}. 

Theorem~\ref{th:key} (iii) shows an interesting non-concentration behavior of $P_{n,m}$ when $n$ is kept small: if the $p_i$'s were i.i.d. uniform on $\set[1]{\frac{\l}{n+1},\l\in \range{n+1}}$ then the histogram $\bm{M}((n+1)\bm{p})$ would follow a multinomial distribution  and the histogram would concentrate around the uniform histogram as $m$ tends to infinity. Rather, the $p_i$'s are here  only exchangeable, not i.i.d., and the histogram does not concentrate when $m$ tends to infinity while $n$ is small. As a case in point, for $n=1$, $M_1((n+1)\bm{p})$ is uniform on $\range{m}$, whatever $m$ is, see \eqref{equunifhisto}. 
Nevertheless, we will show in the next section that a concentration occurs when {\it both $m$  and $n$} tend to infinity. 

 \begin{remark}
  Note that $P^U$ in \eqref{equPU} is the conditional distribution that one would get by applying the de Finetti theorem to the infinite exchangeable sequence $(p_i)_{i\geq 1}$ with  $(p_1,\dots,p_m)\sim P_{n,m}$ for all $m$.
\end{remark}

\begin{figure}[h!]
\begin{center}
\includegraphics[scale=0.35]{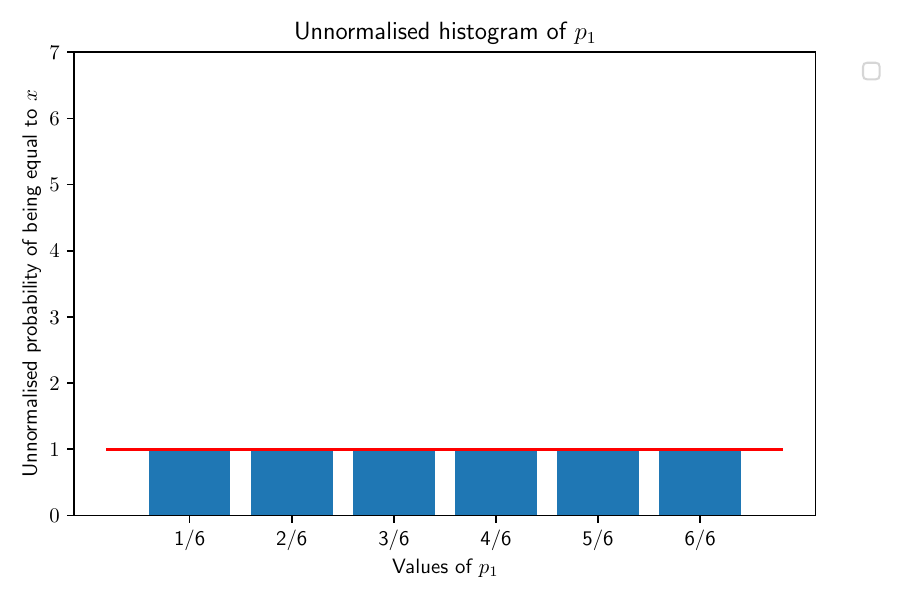}
\includegraphics[scale=0.35]{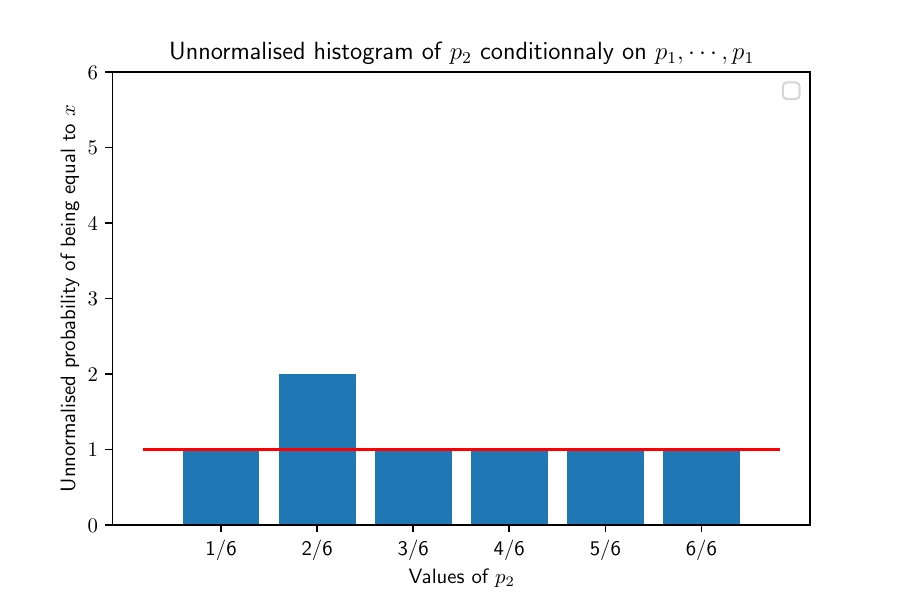}
\includegraphics[scale=0.35]{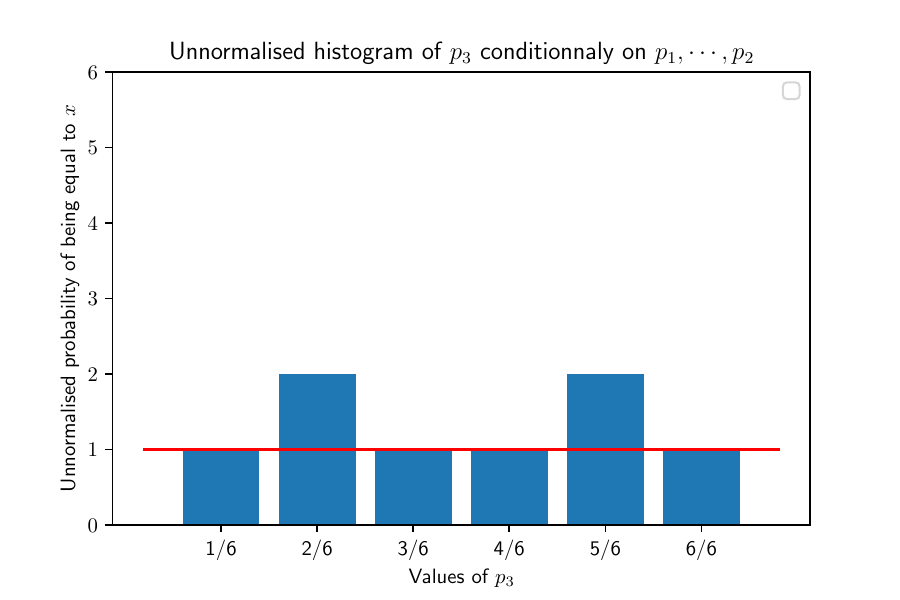}
\includegraphics[scale=0.35]{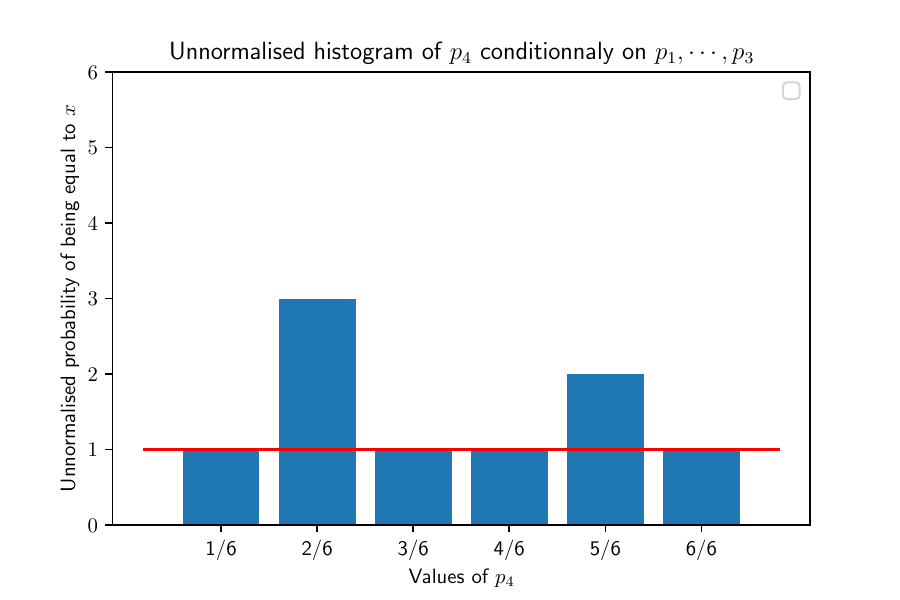}
\includegraphics[scale=0.35]{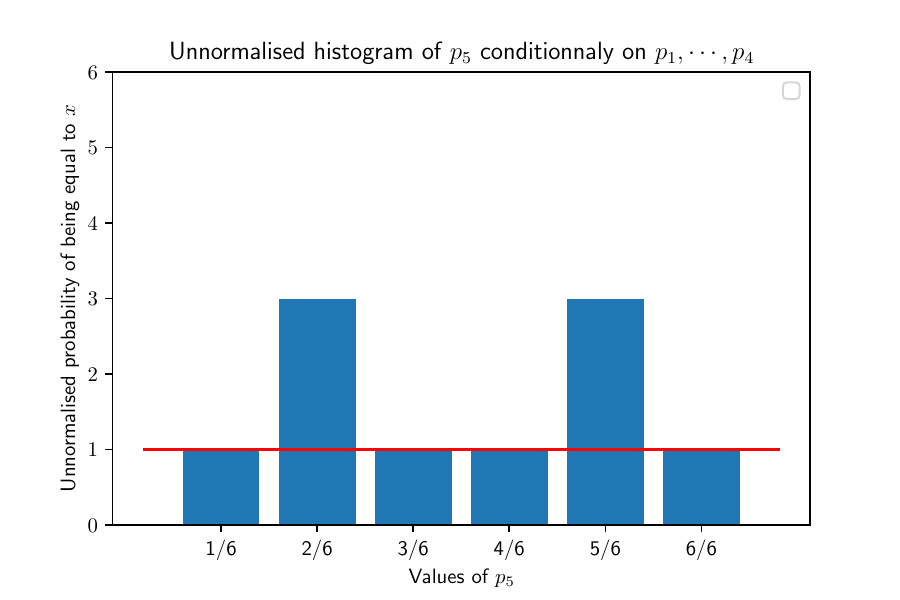}
\includegraphics[scale=0.35]{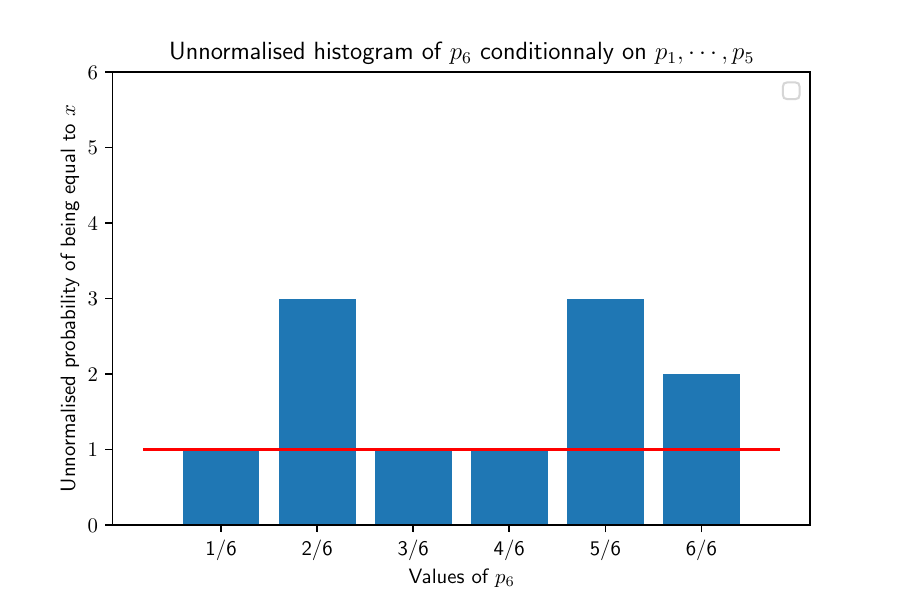}
\end{center}
\caption{Illustration of the sequential realization of $P_{n,m}$ as proved in Theorem~\ref{th:key} (ii) for $n=5$ and $m=6$.}\label{fig:dyn}
\end{figure}

\paragraph{Relation to \cite{f2023universal,huang2023uncertainty}.}
As a consequence of \eqref{equMarkov}, given any {$I\subset \set{\frac{\ell}{n+1}, \ell\in \range{n+1}}$}, we have
\begin{align*}
  \P(p_{i+1}\in {I}\:|\: p_1,\dots,p_{i}) {= \frac{|I|+N_i(I)} 
                                                {n+1+i}
  =\P\left(p_{i+1}\in {I} \:\Big|N_i(I) 
    \right),}
\end{align*}
{where $N_i(I) = |\set{ k \in \range{i}: p_k \in I)}|$.
In words, it means that the P\'olya urn model continues to hold if we group (or ``re-paint'') the
initial $(n+1)$ colors into only two colors, determined by whether the original color label belongs to $I$ or not. }

In particular, we recover the P\'olya urn model put forward by \cite{f2023universal}: letting $Z_i=\ind{p_i>\alpha}$, we have that 
for all $i\in \range{0,m-1}$, 
the distribution of $Z_{i+1}$ conditionally on $Z_1,\dots,Z_{i}$ does not depend on $m$ and is given by 
\begin{equation}
\mathcal{D}(Z_{i+1}\:|\: Z_1,\dots,Z_{i}) =  \frac{\lfloor \alpha(n+1)\rfloor+\sum_{k=1}^{i}\ind{Z_k=j}}{n+1+i}\delta_{0} +\frac{\lceil (1-\alpha)(n+1)\rceil+\sum_{k=1}^{i}\ind{Z_k=j}}{n+1+i}\delta_{1} \label{equMarkov2}.
\end{equation}
Hence, the distribution of $(Z_1,\dots,Z_m)$ corresponds to the distribution of the colors of $m$ successive draws in a standard P\'olya urn model with $2$ colors labeled as $\{0,1\}$ (with an urn starting with $\lfloor \alpha(n+1)\rfloor$ balls $0$ and $\lceil (1-\alpha)(n+1)\rceil$ balls $1$). 

{In particular, we recover} Theorem~1 of \cite{f2023universal} and Theorem~3 in \cite{huang2023uncertainty}.
\begin{corollary}[Theorem~1 in \cite{f2023universal} and Theorem~3 in \cite{huang2023uncertainty}]
In the setting of Theorem~\ref{th:key}, we have for all $\alpha\in (0,1)$ and $k\in \range{m}$, by denoting $k_0=\lceil \alpha(n+1)\rceil$,
\begin{equation}\label{equMarques}
\P\paren{\wh{F}_m(\alpha) = \frac{k}{m}}={m \choose k}  \frac{(n-k_0+1)\dots (n-k_0 + m-k)\times k_0 \dots(k_0+k-1)}{(n+1)\dots (n+m) }.
\end{equation}
\end{corollary}

\begin{proof}
By Proposition~\ref{prop:exch}, \eqref{equ:generation} and the notation of \eqref{equFU}, we have 
\begin{align*}
\P\paren{\wh{F}_m(\alpha) = \frac{k}{m} }&={m \choose k}\E\left[ (U_{(k_0)})^k (1-U_{(k_0)})^{m-k}\right]\\
&={m \choose k} \frac{n!}{(k_0-1)! (n-k_0)!} \int_0^1  u^{k+k_0-1} (1-u)^{m-k+n-k_0} du\\
&={m \choose k} \frac{n!}{(k_0-1)! (n-k_0)!} \frac{(k+k_0-1)!(m+n-k-k_0)!}{(m+n)!},
\end{align*}
by using that $U_{(k_0)}$ follows a beta distribution with parameter $(k_0,n+1-k_0)$  and by using the beta distribution with parameter $(k+k_0,m+n+1-k-k_0)$. This shows the result.
\end{proof}

\section{Numerical bounds and templates}\label{sec:template}

The bound proposed in Theorem~\ref{sec:DKW} are explicit and elegant, but can be conservative in some cases and we develop here the numerical approach mentioned  in Remark~\ref{numispossible}.

We rely on showing \eqref{exactcontrol}, which immediately implies a confidence envelope on $\wh{F}_m$ because 
\begin{align*}
\set{\forall k\in \mathcal{K}\::\: \wh{F}_m(t_k)\leq \frac{k}{m} }&=\set{\forall k\in \mathcal{K}\::\: \wh{F}_m(t_k)< \frac{k+1}{m} }\\
&=\{\forall k\in \mathcal{K}\::\: p_{(k+1)} >t_k \}.
\end{align*}
To establish \eqref{exactcontrol}, we use the notion of template introduced by \cite{BNR2020}, see also \cite{li2022simultaneous}.
A template is a one-parameter family $t_k(\lambda)$, $\lambda\in [0,1]$, $k\in \mathcal{K}\subset \range{m}$, such that $t_k(0) = 0$ and $t_k(\cdot)$ is continuous increasing on $[0, 1]$.   
From above, we have for all $\lambda$,
\begin{align*}
\{\forall k\in \mathcal{K}\::\: \wh{F}_m(t_k(\lambda))\leq k/m \}&=\{\forall k\in \mathcal{K}\::\: p_{(k+1)} >t_k(\lambda) \}\\
&=\Big\{\min_{k\in \mathcal{K}}\{t_k^{-1}(p_{(k+1)})\}>\lambda\Big\}.
\end{align*}
Hence, let us consider 
\begin{equation}\label{equlambdacal}
\lambda(\delta,n,m) = \max\set{\lambda \in \Lambda\::\: \P_{\bm{p}\sim P_{n,m}}\paren[2]{\min_{k\in \mathcal{K}}\{t_k^{-1}(p_{(k)})\}>\lambda}\geq 1-\delta},
\end{equation}
where $\Lambda$ is the finite set $\{t_k^{-1}(\ell/(n+1)),  k\in \mathcal{K}, \ell\in \range{n+1}\}$. 
Then by Proposition~\ref{prop:exch} we have the following result. 

\begin{theorem}\label{thlambdacal} Let us consider the process $\wh{F}_m$ defined by \eqref{equ-ecdfpvalues}, the distribution $P_{n,m}$ given by \eqref{equ:distribution}, a template $t_k(\lambda)$, $\lambda\in [0,1]$, $k\in \mathcal{K}$ as above, and assume \eqref{as:exchangeable} and \eqref{as:noties}. Then we have for all $\delta\in (0,1)$, $n,m\geq 1$,
\begin{align}
\P\paren{\forall k\in \mathcal{K}\::\: \wh{F}_m\paren[2]{t_k\paren[1]{\lambda(\delta,n,m)}}\leq \frac{k}{m}}\geq 1-\delta,
\label{ineqlambdacal}
\end{align}
for $\lambda(\delta,n,m)$ given by \eqref{equlambdacal}.
\end{theorem}

Here are two template choices:
\begin{itemize}
\item The linear template $t_k(\lambda)=k \lambda /m$, $\mathcal{K}=\range{m}$, which leads to the inequality
$$
\P\paren{\exists t \in (0,1) \::\: \wh{F}_m( t)> \frac{\lceil tm/\lambda(\delta,n,m)\rceil}{ m}}\leq \delta,
$$
which recovers the Simes inequality \eqref{Simes} with an adjusted scaling parameter. 
\item The ``beta template'' \cite{BNR2020}, for which $t_k(\lambda)$ is the $\lambda$-quantile of the distribution $\mbox{Beta}(k,m+1-k)$ and thus
$\Lambda=\{F_{\tiny{\mbox{Beta}(k,m+1-k)}}(\ell/(n+1)),  k\in \mathcal{K}, \ell\in \range{n+1}\}$. For instance, it can be used with $\mathcal{K}=\{1+k\lceil \log(m)\rceil ,k\in \range{K}\}$.
\end{itemize}

\section{Proofs}

\subsection{Proof of Proposition~\ref{prop:iid}}\label{proofprop:iid}

Assumption \eqref{as:noties} implies that marginal score distribution is atomless, so that $F$ is
continuous and $1-F(S_i)$ has $\mathrm{Unif}[0,1]$ distribution.
Therefore, $(U_1,\dots,U_{n+m})=(1-F(S_1),\dots,1- F(S_{n+m}))$ are i.i.d. $\sim \mathrm{Unif}[0,1]$. Recall
\[
  p_i=(n+1)^{-1}\paren[3]{1+\sum_{j=1}^n \ind{S_j\geq S_{n+i}}},\:\: i\in \range{m},
\]
since $p_i$ is a function of $S_{n+i}$ and $\dcal$ only, it follows that conditionally on $\dcal$, the variables $p_1,\dots,p_m$ are independent
(and identically distributed).

Since $F$ is continuous, it holds $F^\dagger(F(S_i))=S_i$ almost surely, where $F^\dagger$ is the generalized inverse of $F$. Therefore
$\ind{S_j \geq S_{n+i}} = \ind{U_j \leq U_{n+i}}$ almost surely. Hence, $p_1$ is distributed as 
$$
(n+1)^{-1}\paren[3]{1+\sum_{j=1}^n \ind{U_j\leq U_{n+1}}} =(n+1)^{-1}\paren[3]{1+\sum_{j=1}^n \ind{U_{(j)}\leq U_{n+1}}},
$$
where $U_{(1)}\leq \dots \leq U_{(n)}$ denotes the order statistics of $(U_1,\dots,U_n)$. Therefore, we have for all $x\in [0,1] $,
 \begin{align*}
\P(p_1\leq x\:|\: \dcal)&=\P\paren[3]{1+\sum_{j=1}^n \ind{U_{(j)}\leq U_{n+1}} \leq x(n+1) \Big| \dcal }\\
&= 
\P\paren[3]{1+\sum_{j=1}^n \ind{U_{(j)}\leq U_{n+1}} \leq \lfloor x(n+1) \rfloor \Big| \dcal }\\
&= \P(U_{n+1}<U_{(\lfloor x(n+1) \rfloor)} | \dcal ) = U_{(\lfloor x(n+1) \rfloor)},
 \end{align*}
which finishes the proof.

\subsection{Proof of Proposition~\ref{prop:exch}}\label{proofprop:exch}

If there are no tied scores, which by assumption \eqref{as:noties} happens with probability 1, the ranks $R_i$ of the ordered scores are
well-defined and
the vector $(p_1,\dots,p_{m})$ is only a function of the rank vector $(R_1,\dots,R_{n+m})$.
Namely, $R_i\leq  R_j$ if and only if $S_i\leq S_j$, and the conformal $p$-values \eqref{equemppvalues} can be written as
\begin{equation*}
p_i=(n+1)^{-1}\paren[3]{1+\sum_{j=1}^n \ind{R_j\geq  R_{n+i}}},\:\: i\in \range{m}.
\end{equation*} 
Now, by \eqref{as:exchangeable}, the vector $(R_1,\dots,R_{n+m})$ is uniformly distributed on the permutations of $\range{n+m}$.
Any score distribution satisfying~\eqref{as:noties} and~\eqref{as:exchangeable} therefore gives rise to the same rank distribution,
and thus the same joint $p$-value distribution. This joint distribution has been identified as~\eqref{equ:distribution}-\eqref{equ:generation}
from the result of Proposition~\ref{prop:iid} in the particular case of i.i.d. scores. (Thus the i.i.d. assumption turns out to be unnecessary for what
concerns the joint, unconditional distribution of the $p$-values, but provides a convenient representation.)

\subsection{Proof of Theorem~\ref{th:key}}\label{proofth:key}

\paragraph{Proof of (ii)} 
By \eqref{as:exchangeable},\eqref{as:noties} the permutation that orders the scores $(S_1,\dots,S_{n+m})$ that is $\sigma$ such that 
$$
S_{()}=(S_{\sigma(1)}> \dots >S_{\sigma(n+m)}),
$$
 is uniformly distributed in the set of permutations of $\range{n+m}$. 
 In addition, $\sigma$ is independent of the order statistics $S_{()}$ and we seek for identifying the distribution of $(p_1,\dots,p_m)$ conditionally on $S_{()}$. Next, using again \eqref{as:exchangeable}, we can assume without loss of generality that $j_1\leq \dots\leq j_m$ when computing the probability in \eqref{equjoint}.
 Now, due to the definition \eqref{equemppvalues}, the event $\{(p_1,\dots,p_m)=(j_1/(n+1),\dots,j_m/(n+1)\}$ corresponds to a specific event on $\sigma$. Namely, by denoting $(a_1,\dots,a_\l)$ the vector of unique values of the set $\{j_1,\dots,j_m\}$ with $1\leq a_1<\dots<a_\l\leq n$, and $M_k=\sum_{i=1}^m \ind{j_i=a_k}$, $1\leq k\leq \ell$, the corresponding multiplicities, the above event corresponds to the situation
 \begin{align*}
&\underbrace{S_{\sigma(1)}>\cdots>S_{\sigma( a_1-1)}}_{ a_1-1 \text{ null scores}}>\underbrace{S_{\sigma(a_1)}>\cdots>S_{\sigma(a_1+M_1-1)}}_{M_1\text{ test scores in $\{S_{n+1},\dots,S_{n+M_1}\}$}}>\\
&\underbrace{S_{\sigma( a_1+M_1)}>\cdots>S_{\sigma( a_2+M_1-1)}}_{\text{$ a_2- a_1$ null scores}}>\underbrace{S_{\sigma(a_2+M_1)}>\cdots>S_{\sigma_1(a_2+M_1+M_2-1)}}_{M_2\text{ test scores in $\{S_{n+M_1+1},\dots,S_{n+M_1+M_2}\}$}}>\cdots\\
&\underbrace{S_{\sigma( a_{\ell-1}+M_1+\cdots+M_{\ell-1})}>\cdots>S_{\sigma( a_\ell+M_1+\cdots+M_{\ell-1}-1)}}_{\text{$ a_\ell- a_{\ell-1}$  null scores}}>\underbrace{S_{\sigma(a_\ell+M_1+\cdots+M_{\ell-1})}>\cdots>S_{\sigma(a_\ell+m-1)}}_{M_\ell\text{ test scores in $\{S_{ n+M_1+\cdots+M_{\ell-1}+1},\dots,S_{n+m}\}$}}>\\
&\underbrace{S_{( a_\ell+m)}>\cdots>S_{(n+m)}}_{n- a_\ell+1\text{ null scores}}.
\end{align*}
This event can be formally described as follows:
 \begin{align*}
  \Big\{\forall k \in \range{\ell}\::\: &\sigma\paren[1]{\left\{a_\ell+M_1+\cdots+M_{k-1},\dots,a_\ell+M_1+\cdots+M_{k}-1\right\}} \\
& = \left\{n+M_1+\cdots+M_{k-1}+1,\dots,n+M_1+\cdots+M_{k}\right\}\Big\}.
 \end{align*}
Since $\sigma$ is uniformly distributed in the set of permutations of $\range{n+m}$, the probability of this event (conditionally on $S_{()}$) is equal to
$
n! \big(\prod_{k=1}^\ell (M_k!)\big)/(n+m)!,
$
which yields  \eqref{equjoint}.

\paragraph{Proof of (i)} 
By using \eqref{equjoint} of (ii), we have 
$$
\P(p_{i+1}=j_{i+1}/(n+1)\:|\: (p_1,\dots,p_i)=(j_1/(n+1),\dots,j_i/(n+1))) = \frac{ \bm{M}(j_1,\dots,j_{i+1})! \frac{n!}{(n+i+1)!}}{ \bm{M}(j_1,\dots,j_i)! \frac{n!}{(n+i)!}}.
$$
Now, we have 
 \begin{align*}
\bm{M}(j_1,\dots,j_{i+1})!&=\prod_{j=1}^{n+1} \left[\left(\sum_{k=1}^{i+1} \ind{j_k=j}\right)!\right]=\prod_{j=1}^{n+1} \left[\left(\sum_{k=1}^{i} \ind{j_k=j}+\ind{j_{i+1}=j}\right)!\right]\\
&=\prod_{j=1}^{n+1} \left(\sum_{k=1}^{i} \ind{j_k=j}\right)! \left[1+\ind{j_{i+1}=j}\sum_{k=1}^{i} \ind{j_k=j}\right]\\
&=\bm{M}(j_1,\dots,j_{i})!\left[1+\ind{j_{i+1}=j}\sum_{k=1}^{i} \ind{j_k=j}\right].
 \end{align*}
This proves \eqref{equMarkov}.

\paragraph{Proof of (iii)} 
For all $\bm{m}=(m_1,\dots,m_{n+1})\in \range{0,m}^{n+1}$ with $ m_1+\dots+m_{n+1}=m$,  we have
\begin{align*}
\P\big(\bm{M}((n+1)\bm{p})=\bm{m}\big) &=\sum_{\bm{j}\in \range{n+1}^m} \ind{\bm{M}(\bm{j})=\bm{m}} \P\big((n+1)\bm{p}=\bm{j}\big)\\
&=\bm{m}! \frac{n!}{(n+m)!}\sum_{\bm{j}\in \range{n+1}^m} \ind{\bm{M}(\bm{j})=\bm{m}}  \\
&=\bm{m}! \frac{n!}{(n+m)!} \frac{m!}{\bm{m}!}= \frac{n!m!}{(n+m)!}, 
\end{align*}
where we have used (ii) and the multinomial coefficient.

\subsection{Proof of Theorem~\ref{thDKW}}\label{sec:thDKWproof}

First observe that the LHS of \eqref{boundDKWup} is $0$ if $\lambda\geq 1$ so that we can assume $\lambda<1$.

Let us  prove \eqref{boundDKWup} with the more complex bound
\begin{align}
B^{\mbox{\tiny DKWfull}}(\lambda,n,m)&:=
\frac{n}{n+m} e^{-2m \lambda^2} + \frac{m}{n+m} e^{-2n\lambda^2} +  C_{\lambda,n,m}\frac{2\sqrt{2\pi}\lambda nm}{(n+m)^{3/2}}  e^{-\frac{2nm}{n+m}\lambda^2},
\label{BDKWfull}
\end{align}
where $C_{\lambda,n,m}=\P(\mathcal{N}(\lambda \mu,\sigma^2)\in [0,\lambda])<1$, for $\sigma^2=(4(n+m))^{-1}$ and $\mu=n(n+m)^{-1}$. 
Let us comment the expression \eqref{BDKWfull} of $B^{\mbox{\tiny DKWfull}}(\lambda,n,m)$. As we can see, the role of  $n$ and $m$ are symmetric (except in $C_{\lambda,n,m}$, that we can always further upper-bound by $1$), and the two first terms are a convex combination of the usual DKW bounds for $m$ and $n$ i.i.d. variables, respectively. The third term is a ``crossed'' term between $n$ and $m$, which becomes negligible if $n\gg m$ or $n\ll m$ but should be taken into account otherwise.

Below, we establish 
\begin{align}
\P\paren[3]{\sup_{t\in [0,1]}\paren[1]{\wh{F}_m(t) - I_n(t)} > \lambda}&\leq B^{\mbox{\tiny DKWfull}}(\lambda,n,m);\label{boundDKWfullup}\\
\P\paren[3]{\sup_{t\in [0,1]}\paren[1]{-\wh{F}_m(t) + I_n(t)} > \lambda}&\leq B^{\mbox{\tiny DKWfull}}(\lambda,n,m);\label{boundDKWfulldown}\\
\P\paren{\norm[1]{\wh{F}_m - I_n}_\infty > \lambda}&\leq  2  B^{\mbox{\tiny DKW}}(\lambda,n,m).\label{boundDKWfulltwosided}
\end{align}
The result will be proved from \eqref{boundDKWfullup} because $B^{\mbox{\tiny DKWfull}}(\lambda,n,m)\leq B^{\mbox{\tiny DKW}}(\lambda,n,m)$ since $n\vee m\geq nm/(n+m)$ and $C_{\lambda,n,m}\leq 1$.

The proof relies on Proposition~\ref{prop:exch} and the representation~\eqref{equ:generation}. Let $U=(U_1,\dots,U_n)$ i.i.d. $\sim U(0,1)$,  and denote $F^U(x)=U_{(\lfloor (n+1)x\rfloor)}$, $x\in [0,1]$. Conditionally on $U$, draw $(q_i(U),i\in \range{m})$  i.i.d. of common c.d.f. $F^U$ and let
$$
\wh{G}_m(t)=m^{-1}\sum_{i=1}^m \ind{q_i(U)\leq t}, \:\:\: t\in [0,1],
$$
the empirical c.d.f. of $(q_i(U),i\in \range{m})$. By Proposition~\ref{prop:exch}, we have that
{$\wh{F}_m$ has the same distribution as $\wh{G}_m$ (unconditionally on $U$)}, so that for any fixed $n,m\geq 1$ and $\lambda>0$,
\begin{align}\label{equinterm}
\P\paren[3]{\sup_{t\in [0,1]}\paren[1]{\wh{F}_m(t) - I_n(t)} > \lambda}&=\e{\P\paren[3]{ \sup_{t\in [0,1]}\paren[1]{\wh{G}_m(t) - I_n(t)} > \lambda\:\Big|\: U}}.
\end{align}

We now prove the bound \eqref{boundDKWfullup} (the proof for \eqref{boundDKWfulldown} is analogous).
Denote $Z=\sup_{t\in [0,1]}\paren[1]{F^U(t)-I_n(t)}\in [0,1]$. 
We write by \eqref{equinterm} and the triangle inequality
\begin{align*}
\P\paren{\sup_{t\in [0,1]}\paren[1]{\wh{F}_m(t) - I_n(t)} > \lambda} &\leq \e{ \P\Big(\sup_{t\in [0,1]}(\wh{F}_m(t) - F^U(t)) + Z  > \lambda\:\Big|\: U\Big)}
\\
&\leq \e{ \P\Big(\sup_{t\in [0,1]}(\wh{F}_m(t) - F^U(t))  \geq  \big(\lambda- Z\big)_+\:\Big|\: U\Big)}\\
&\leq \e{ e^{- 2m  (\lambda- Z)^2_+}}.
\end{align*}
{The last inequality above is the DKW inequality~\citep{Mass1990} applied to control the inner conditional probability, since
  conditionally to $U$, $\wh{F}_m$ is the e.c.d.f. of $(q_i(U),\in \range{m})$, which are i.i.d. $\sim F^U$; and $Z$ conditional to $U$ is a constant.}
Now the last bound can be rewritten as
\begin{align}
  \int_0^{1} \P\paren{e^{- 2m  (\lambda- Z)^2_+}> v} dv
  &= e^{-2m \lambda^2} +  \int_{e^{-2m \lambda^2}}^1\P\paren{(\lambda- Z)_+< \sqrt{\log(1/v)/(2m)}} dv \notag\\
&= e^{-2m \lambda^2} +  \int_{e^{-2m \lambda^2}}^1 \P\paren{\lambda- Z< \sqrt{\log(1/v)/(2m)}} dv \notag \\
&= e^{-2m \lambda^2} +  \int_{e^{-2m \lambda^2}}^1 \P\paren{Z> \paren[1]{\lambda-\sqrt{\log(1/v)/(2m)}}} dv. \label{eq:intermed}
\end{align}
To upper bound the integrand above, denote $\wh{H}_n$ the ecdf of $(U_1,\ldots,U_n)$; it holds for any $x \in [0,1]$:
\begin{align*}
  \P\paren{Z>x} 
& = \P\paren[3]{\sup_{t\in [0,1]}\left(U_{(\lfloor (n+1)t\rfloor)}-\lfloor (n+1)t\rfloor/(n+1) \right)> x}\\
&=\P\left(\exists k\in \range{n}\::\: U_{(k)} > x+k/(n+1)\right)\\
&=\P\paren[3]{\exists k\in \range{n}\::\: \sum_{i=1}^{n}\ind{U_i\leq x+k/(n+1)} \leq k-1}\\
&=\P\left(\exists k\in \range{n}\::\:  \wh{H}_n\paren[1]{x+k/(n+1)} -  [x+k/(n+1)] \leq (k-1)/n -  [x+k/(n+1)] \right).\\
&\leq P\left(\exists k\in \range{n}\::\:  \wh{H}_n\paren[1]{x+k/(n+1)} -  [x+k/(n+1)] \leq -  x \right)\\
& \leq e^{-2n x^2 },
\end{align*}
where we used $(k-1)/n\leq k/(n+1)$ in the first inequality,
and the left-tail DKW inequality for the last one. Plugging this into~\eqref{eq:intermed} yields
\begin{align*}
  \int_0^{1} \P(e^{- 2m  (\lambda- Z)^2_+}> v) dv
  &\leq e^{-2m \lambda^2} +  \int_{e^{-2m \lambda^2}}^1 e^{-2 n (\lambda-\sqrt{\log(1/v)/(2m)})^2}dv.
\end{align*}
Now letting $u=\sqrt{\log(1/v)/(2m)}$ (hence $v=e^{-2m u^2}$, $dv=-4mu e^{-2m u^2}du$), we obtain
\begin{align*}
\P\paren[3]{\sup_{t\in [0,1]}\paren[1]{\wh{F}_m(t) - I_n(t)} > \lambda}&\leq e^{-2m \lambda^2} +4 m \int_0^\lambda ue^{-2 n (\lambda-u)^2} e^{-2m u^2}du.
\end{align*}
Now, 
by denoting $\sigma^2=\paren[1]{4(n+m)}^{-1}$ and $\mu=n(n+m)^{-1}$, we get
\begin{align*}
e^{\frac{2nm}{n+m}\lambda^2} \int_0^\lambda ue^{-2 n (\lambda-u)^2} e^{-2m u^2}du 
&= \int_0^\lambda ue^{-2(n+m)\left(u-\frac{n\lambda}{n+m}\right)^2 du}\\
&=\int_0^\lambda ue^{-\frac{1}{2\sigma^2}\left (u-\lambda\mu\right )^2}du\\
&=\int_0^\lambda (u-\lambda\mu)e^{-\frac{1}{2\sigma^2}\left (u-\lambda\mu\right )^2}du+\int_0^\lambda \lambda\mu e^{-\frac{1}{2\sigma^2}\left (u-\lambda\mu\right )^2}du\\
&= \sigma^2e^{-2\lambda^2\frac{n^2}{m+n}}-\sigma^2e^{-2\lambda^2\frac{m^2}{m+n}}+\lambda\mu 
\sqrt{2\pi} \sigma C_{\lambda,n,m}.
\end{align*}
where $C_{\lambda,n,m}=\P\paren[1]{\mathcal{N}(\lambda \mu,\sigma^2)\in [0,\lambda]}$.
Hence, 
\begin{align*}
 \int_0^\lambda ue^{-2 n (\lambda-u)^2} e^{-2m u^2}du 
&= 
e^{-\frac{2nm}{n+m}\lambda^2}\left(\sigma^2e^{-2\lambda^2\frac{n^2}{m+n}}-\sigma^2e^{-2\lambda^2\frac{m^2}{m+n}}+\lambda\mu 
\sqrt{2\pi} \sigma C_{\lambda,n,m}\right)\\
&=\sigma^2e^{-2n\lambda^2}-\sigma^2 e^{-2m\lambda^2}+\lambda\mu 
\sqrt{2\pi} \sigma C_{\lambda,n,m} e^{-\frac{2nm}{n+m}\lambda^2}.
\end{align*}
This leads to 
\begin{align*}
&e^{-2m \lambda^2}+4 m \int_0^\lambda ue^{-2 n (\lambda-u)^2} e^{-2m u^2}du\\
&=\frac{n}{n+m} e^{-2m \lambda^2} + \frac{m}{n+m}e^{-2n\lambda^2} + \lambda\sqrt{2\pi} \frac{nm}{(n+m)^{3/2}} 2C_{\lambda,n,m}e^{-\frac{2nm}{n+m}\lambda^2},
\end{align*}
which finishes the proof of \eqref{boundDKWfullup}.

Finally, let us prove $B^{\mbox{\tiny DKW}}(\lambda^{\mbox{\tiny DKW}}_{\delta,n,m},n,m)\leq \delta$ for $\lambda^{\mbox{\tiny DKW}}_{\delta,n,m}=\Psi^{(r)}(1)$ where $\Psi^{(r)}$ denotes the function $\Psi$ iterated $r$ times (for an arbitrary integer $r\geq 1$), where
$$
\Psi(x)=1\wedge \wt{\Psi}(x); \;\;\; \wt{\Psi}(x) := \paren{ \frac{\log(1/\delta)+\log\big(1+ \sqrt{2\pi}  \frac{2\tau_{n,m}x}{(n+m)^{1/2}}\big)}{2\tau_{n,m}} }^{1/2}.
$$
If $\Psi(1) = 1$, then $\Psi^{(r)}(1) = 1$ for all $r$ and the announced claim holds since
$B^{\mbox{\tiny DKW}}(1,n,m)=0$ by definition. We therefore assume $\Psi(1) <1$ from now on.
Since $\Psi$ is non-decreasing, by an immediate recursion we have $\Psi^{(r+1)}(1) \leq \Psi^{(r)}(1) <1 $,
for all integers $r$.

On the other hand, note that for any $x\in(0,1)$ satisfying $\Psi(x) \leq x <1$, it holds $\Psi(x) = \wt{\Psi}(x)$ and thus
\[
  B^{\mbox{\tiny DKW}}(\Psi(x),n,m) = \brac{1+\frac{2\sqrt{2\pi}\Psi(x) \tau_{n,m}}{(n+m)^{1/2}}}
                         \brac{1+\frac{2\sqrt{2\pi}x \tau_{n,m}}{(n+m)^{1/2}}}^{-1} \delta \leq \delta.
\]
Since we established that $x=\Psi^{(r)}(1)$ satisfies $\Psi(x) \leq x$ for any integer $r$ the claim follows.

\section{Explicit control of \eqref{controlalphaFWER}}
\label{sec:minp}

By applying \eqref{equMarques} with $k=0$, the control \eqref{controlalphaFWER} for $\ol{\alpha}=0$ is satisfied by choosing
$$
t_{0,\delta}=\max\set{k/(n+1)\::\:  \frac{(n-k+1)\dots (n-k + m)}{(n+1)\dots (n+m) } \geq 1-\delta, k\in \range{n+1}}.
$$
{We can also obtain an implicit formula for $t_{\ol{\alpha},\delta}$ when $\ol{\alpha}>0$ as follows. By definition, $t_{\ol{\alpha},\delta}$ is the maximum of the $t\in \range{n+1}/(n+1)$ such that $\P_{\bm{p}\sim P_{n,m}}( p_{(\lfloor \ol{\alpha} m\rfloor+1)}\leq t)  \leq \delta$, or equivalently $\P_{\bm{p}\sim P_{n,m}}( \FCP(\bm{\mathcal{C}}(t))> \ol{\alpha})\leq \delta$.  
The latter probability can be obtained explicitly from \eqref{equMarques} with the formula
\begin{align*}
\P_{\bm{p}\sim P_{n,m}}( \FCP(\bm{\mathcal{C}}(t))> \ol{\alpha}) = \sum_{k=\lfloor \ol{\alpha}m\rfloor+1}^m  {m \choose k}  \frac{(n-k_0+1)\dots (n-k_0 + m-k)\times k_0 \dots(k_0+k-1)}{(n+1)\dots (n+m) },
\end{align*}
where $k_0=\lceil t(n+1)\rceil$.
Of course whenever this formula is too computationally complex for a practical use (e.g., when $m$ is large), we can alternative use a Monte-Carlo scheme to simulate draws from $P_{n,m}$ and thus approximate $t_{\alpha,\delta}$ as an empirical $\delta$-quantile of $p_{(\lfloor \ol{\alpha} m\rfloor+1)}$ with $\mathbf{p}\sim P_{n,m}$.}

\section{On the tightness of the new DKW bound}\label{sec:tighnessDKW}

{The bound $B^{\mbox{\tiny DKW}}(\lambda,n,m)$ \eqref{BDKW} is simple and explicit and we comment here briefly about its sharpness:
\begin{itemize}
\item First,  for a fixed $m$ we have $B^{\mbox{\tiny DKW}}(\lambda,n,m)\to  \bm{1}_{\set{\lambda <1}} e^{-2m \lambda^2}$ when $n$ tends to infinity. This bound hence recovers the usual DKW  inequality \cite{Mass1990} for $\wh{F}_m$, which is well expected because $n=\infty$ corresponds to the case of i.i.d. uniform $p$-values ('theoretical' $p$-values rather than 'conformal' $p$-values). In addition, note that the usual DKW bound (in $n$) can be also recovered when $n$ is fixed and $m$ tends to infinity.
\item This bound provides the correct variance term. Indeed, we can deduce from Theorem~\ref{th:key} the following equality: for all $t,s\in \R$,
\[ \Cov_{\mathbf{p}\sim P_{n,m}}\left (\wh{F}_m(t),\wh{F}_m(s)\right )=\frac{m+n+1}{m(n+2)}\left ({I_n(t)\wedge I_n(s)-I_n(t) I_n(s)}\right ).
\] 
Clearly, we have $\frac{m+n+1}{m(n+2)}\sim {\tau_{n,m}}$ when $m\wedge n\rightarrow +\infty$.
\item The bound is compared to the true probability by using simulations in Figure~\ref{fig:DKWTight}. We observe that the bound is fairly close to the target when $\lambda$ is large enough or/and $m\wedge n$ is large.
\end{itemize}
As mentioned in Remark~\ref{numispossible}, recall that this bound can be made sharper by using (non-explicit) numerical approximations.  
}

\begin{figure}[h!]

\includegraphics[scale=0.3]{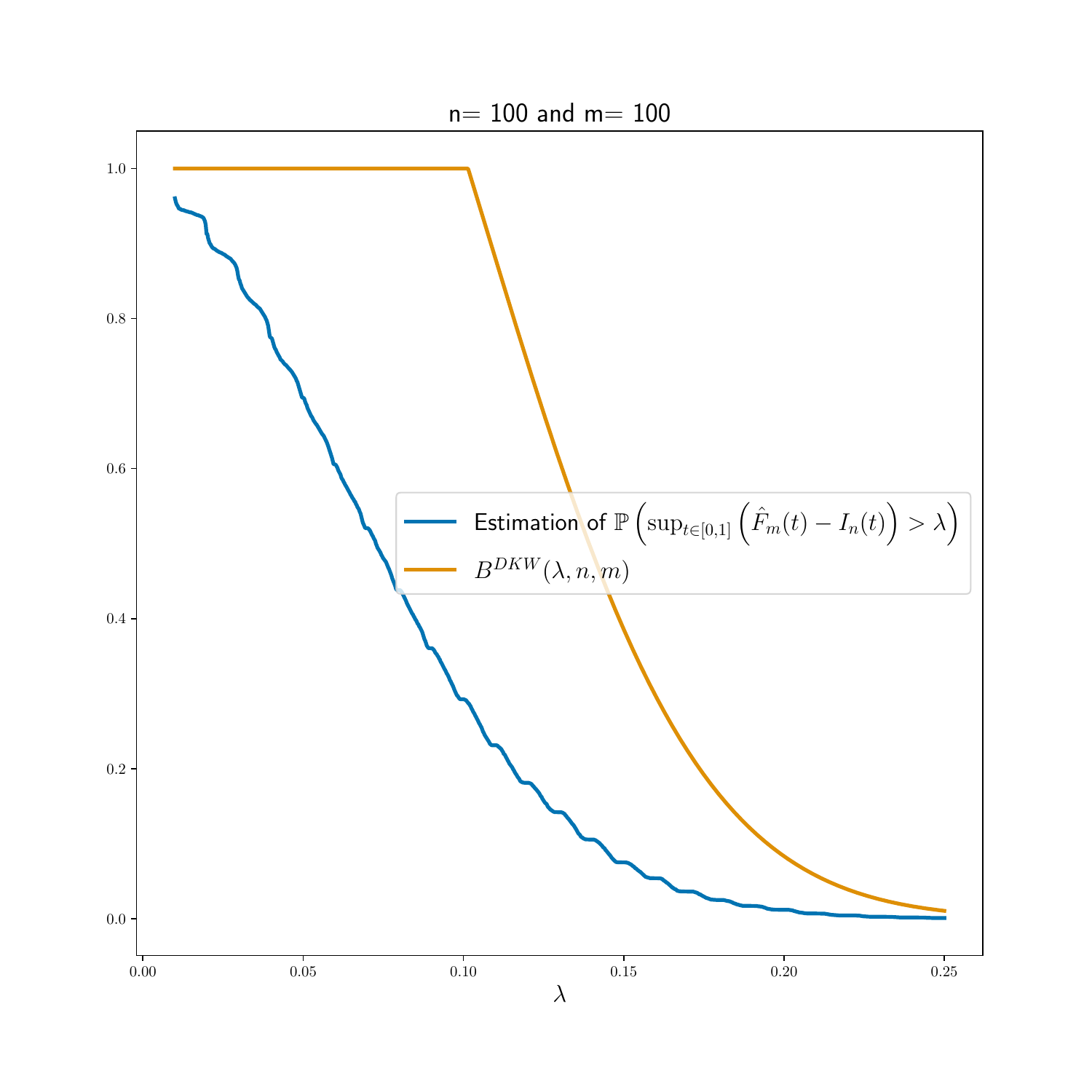}
\includegraphics[scale=0.3]{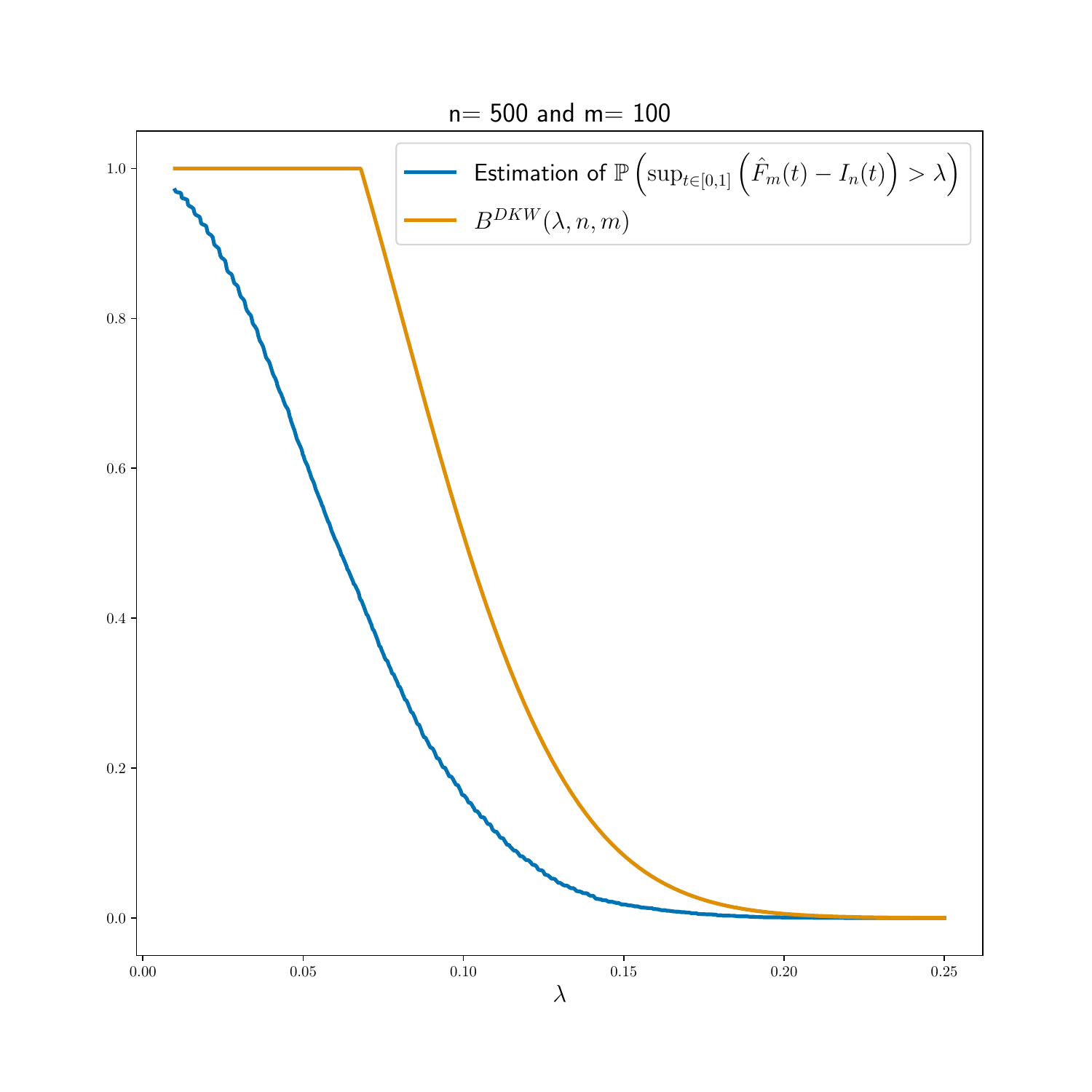}
\includegraphics[scale=0.3]{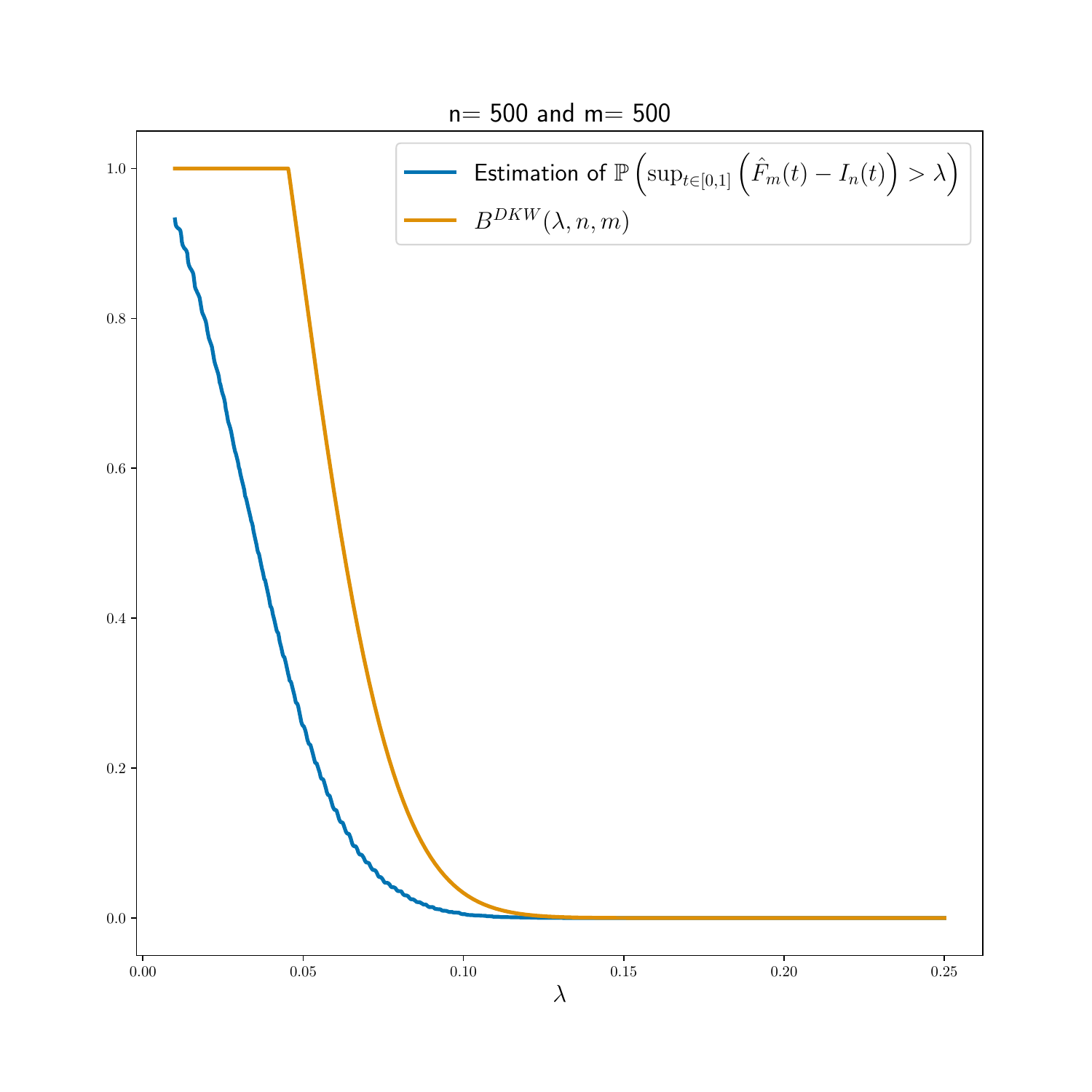}
\includegraphics[scale=0.3]{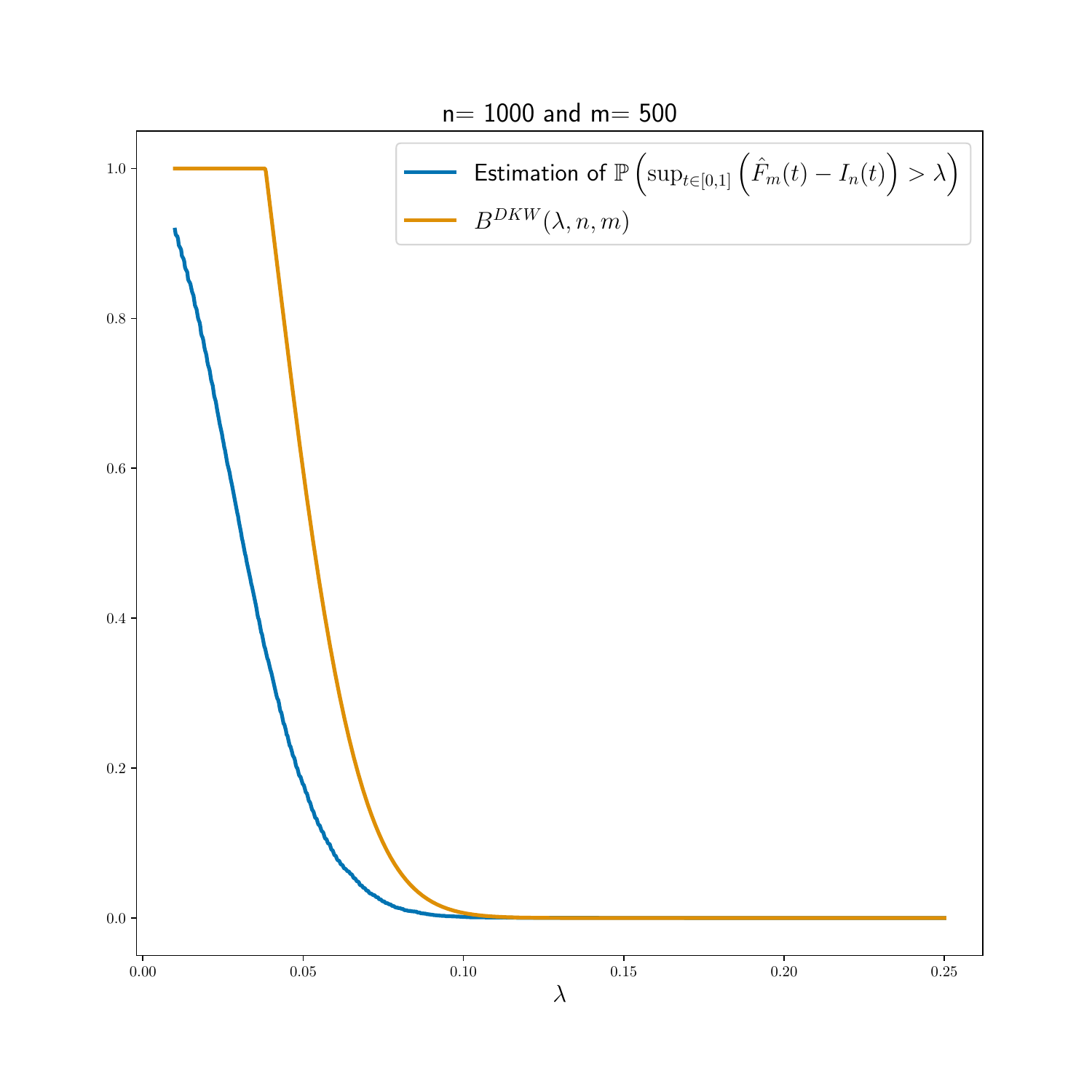}

\caption{
Plot of  $\lambda\mapsto \P(\sup_{t\in[0,1]}(\wh{F}_m(t)-I_n(t))>\lambda)$ (Blue) and of $\lambda\mapsto B^{\mbox{\tiny DKW}}(\lambda,n,m)$ (Orange) for different values of $n$ and $m$. These probabilities are estimated with $10^4$ Monte-Carlo iterations.}\label{fig:DKWTight}
\end{figure}

\section{Proof of Corollary~\ref{cor:ThresholdFDP}}\label{sec:adadetect}

Let $m_0=|\cH_0|$. We establish the following more general result.

\begin{lemma}\label{lem:interm}
With probability at least $1-\delta$, we have both
\begin{align}
& \forall t\in (0,1), \:\:\FDP(\mathcal{R}(t))\leq 
 \frac{ m_0 I_n(t)+  m_0 \lambda^{\mbox{\tiny DKW}}_{\delta,n,m_0}}{1\vee |\mathcal{R}(t)|} \label{ThresholdFDPboundDKW2};\\
 & m_0\leq \max\set{r\in \range{m}\::\:\inf_t \paren{\frac{\sum_{i=1}^m\ind{p_i> t} + \max_{u\in\range{r}}\Big(u \lambda^{\mbox{\tiny DKW}}_{\delta,n,u}\Big) }{1-I_{n}(t)}}\geq r}\label{ThresholdFDPboundDKW3}.
\end{align}
\end{lemma}

Lemma~\ref{lem:interm} implies Corollary~\ref{cor:ThresholdFDP} because if $\hat{m}_0$ is as in \eqref{m0hat}, with probability at least $1-\delta$, $\hat{m}_0\geq {m}_0$ by \eqref{ThresholdFDPboundDKW3}, and by \eqref{ThresholdFDPboundDKW2}
$$
\forall t\in (0,1),\:\:\FDP(\mathcal{R}(t))\leq  \frac{ m_0 I_n(t)+  m_0 \lambda^{\mbox{\tiny DKW}}_{\delta,n,m_0}}{1\vee |\mathcal{R}(t)|}\leq    \frac{ \hat{m}_0 I_n(t)+  \max_{r\in\range{\hat{m}_0}}\Big(r \lambda^{\mbox{\tiny DKW}}_{\delta,n,r}\Big)}{1\vee |\mathcal{R}(t)|}.
$$

Now, let us prove Lemma~\ref{lem:interm}.

First, in the work of \cite{marandon2022machine}, it is proved that $(S_{1}, \ldots, S_{n}, S_{n+i},i\in \cH_0)$ is exchangeable conditionally on $(S_{n+i},i\in \cH_1)$ (see Lemma~3.2 therein).
Hence, the vector $(S_{1}, \ldots, S_{n}, S_{n+i},i\in \cH_0)$, of size $n+m_0$, and the $p$-value vector $(p_i,i\in \cH_0)$, of size $m_0$, fall into the setting described in Section~\ref{sec:setting} with calibration scores being $(S_{i})_{i\in \range{n}}$ and test scores being $(S_{n+i})_{i\in \cH_0}$. By Proposition~\ref{prop:exch}, this means $(p_i,i\in \cH_0)\sim P_{n,m_0}$.

Second, consider the event
\begin{equation}\label{eqOmega}
\Omega=\set{\sup_{t\in [0,1]}(\wh{F}_{ m_0}(t) - I_{n}(t)) \leq \lambda^{\mbox{\tiny DKW}}_{\delta,n, m_0}}.
\end{equation}
By applying Theorem~\ref{thDKW} and the explicit bound \eqref{boundDKWupexplicit}, we have  $\P(\Omega)\geq 1-\delta$. 
Next, $|\mathcal{R}(t)\cap \cH_0|=m_0\wh{F}_{m_0}(t)\leq m_0 I_{n}(t) +  m_0 \lambda^{\mbox{\tiny DKW}}_{\delta,n, m_0}$ on $\Omega$. This gives  \eqref{ThresholdFDPboundDKW2}.

Let us now turn to prove \eqref{ThresholdFDPboundDKW3} on $\Omega$.
For this, let us observe that on this event, we have for all $t\in(0,1)$,
\begin{align*}
\sum_{i=1}^m\ind{p_i> t}\geq \sum_{i\in \cH_0}\ind{p_i> t} &=  m_0(1-\wh{F}_{ m_0}(t)) \\
&\geq   m_0(1-I_{n}(t)) -  \max_{r\in\range{m_0}} \Big(r  \lambda^{\mbox{\tiny DKW}}_{\delta,n,r}\Big) \\
\end{align*}
Hence, $m_0$ is an integer $r\in \range{m}$ such that $\inf_t \Big(\frac{\sum_{i=1}^m\ind{p_i> t} +\max_{u\in\range{r}}\{ u \lambda^{\mbox{\tiny DKW}}_{\delta,n,u} \}}{1-I_{n}(t)}\Big)\geq r$, which gives \eqref{ThresholdFDPboundDKW3}.

\section{Confidence envelope and bounds derived from the Simes inequality}\label{sec:Simes}

As proved in \cite{bates2023testing} in the i.i.d. case, and since the joint distribution of the conformal $p$-values is the same under exchangeability of the scores (Proposition~\ref{prop:exch}), the conformal $p$-values are positively regressively dependent on each one of a subset (PRDS)  under \eqref{as:exchangeable} and \eqref{as:noties}, see \cite{BY2001} for a formal definition of the latter.

\begin{figure}[h!]
\begin{center}
\includegraphics[scale=0.4]{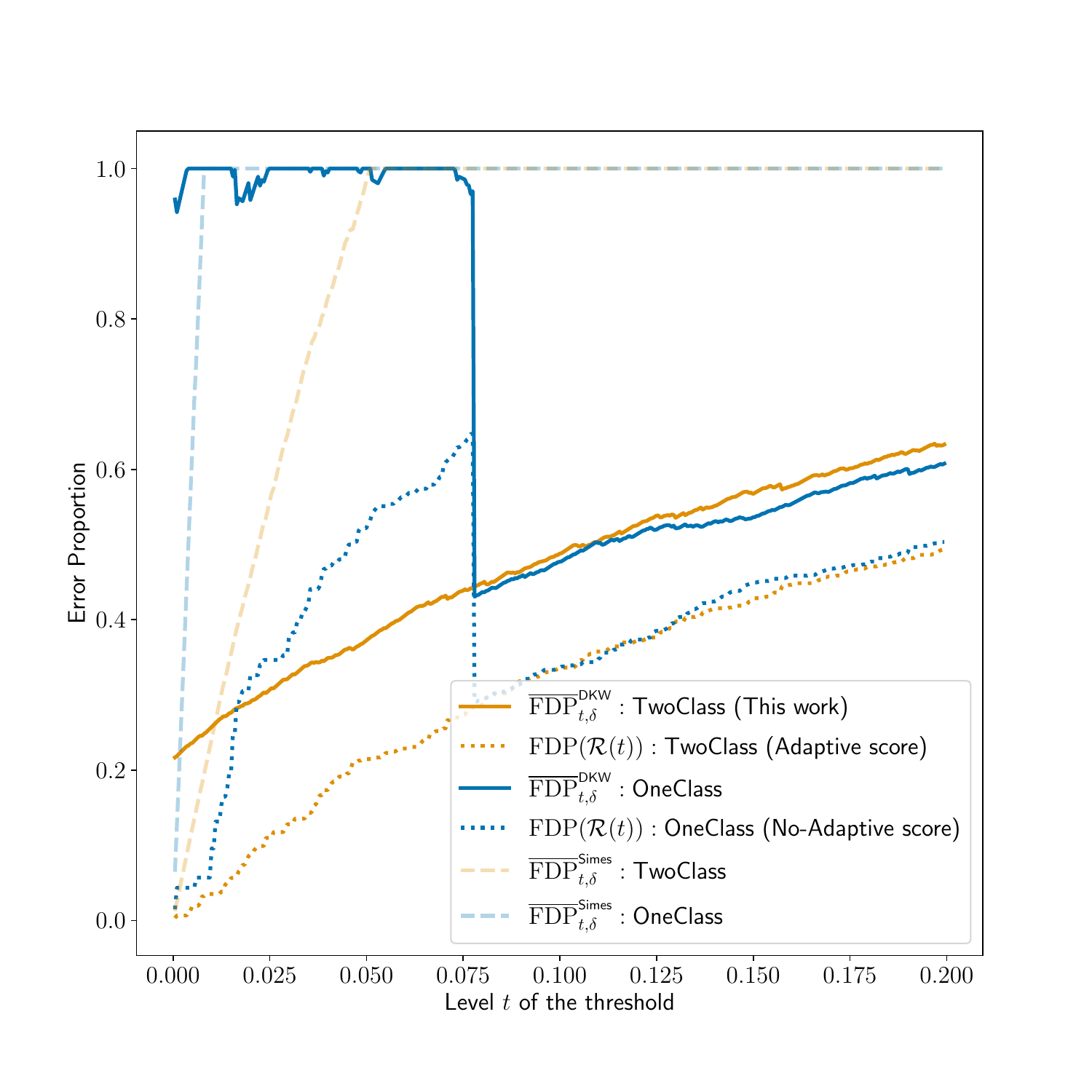}
\end{center}
\caption{
Same as Figure~\ref{fig:nd} with in addition Simes bound ${\overline{\FDP}}^{\mbox{\tiny Simes}}_{t,\delta}$  \eqref{AdaFDPboundSimesb} (transparent dashed, $\delta=0.2$).
\label{fig:ndSimes}}
\end{figure}

Hence, by \cite{BY2001}, the Simes inequality \citep{Sim1986} is valid, that is, for all $\lambda>0$, we have
\begin{align}\label{Simes}
\P\paren{\sup_{t\in (0,1]}(\wh{F}_m(t)/t) \geq \lambda}&\leq 1/\lambda. 
\end{align}
This envelope can be applied in the two applications of the paper as follows:
\begin{itemize}
\item[(PI)] Under the condition of Corollary~\ref{corsimultaneous}, the bound
\begin{align}
{\overline{\FCP}}^{\mbox{\tiny Simes}}_{\alpha,\delta}&=  (\alpha/\delta) \ind{\alpha\geq 1/(n+1)}\label{boundfalsepositiveSimes}
\end{align}
is valid for \eqref{controlalphaFWERunif}.
\item[(ND)] Under the condition of Corollary~\ref{cor:ThresholdFDP} the following control is valid
\begin{equation}\label{AdaFDPboundSimesa}
\P\Big( \forall t\in (0,1),\FDP(\mathcal{R}(t))\leq {\overline{\FDP}}^{\tiny \mbox{Simes}}_{t,\delta}\Big)\geq 1-\delta,
\end{equation}
for
\begin{align}
 \ol{\FDP}^{\mbox{\tiny Simes}}_{t,\delta} :=\frac{\hat{m}_0 t/\delta}{1\vee |\mathcal{R}(t)|}  ,\label{AdaFDPboundSimesb}
\end{align}
for any estimator 
\begin{align}
\hat{m}_0\geq m\wedge \inf_{t\in (0,\delta)} \frac{\sum_{i=1}^m \ind{p_i>t}}{1-t/\delta}.
\label{m0chapSimes}
\end{align}
\end{itemize}

A comparison between the Simes bound and the DKW bound is presented in Figure~\ref{fig:ndSimes} for the (ND) task. While the Simes bound is better for extremely small $t$, the DKW bound is in general sharper.

\section{Uniform FDP bound for AdaDetect}\label{sec:AdaDetect}

AdaDetect \citep{marandon2022machine} is obtained by applying the Benjamini-Hochberg (BH) procedure \citep{BH1995} to the conformal $p$-values, that is, $\AdaDetect_\alpha:=\mathcal{R}(\alpha\hat{k}_\alpha/m)$, where
\begin{equation}\label{equkchapeau}
\hat{k}_\alpha := \max\set{k\in \range{0,m}\::\: \sum_{i=1}^m \ind{p_i\leq \alpha k/m}\geq k}.
\end{equation}
It is proved there to control the false discovery rate (FDR), defined as the mean of the FDP:  
\begin{equation}\label{fdrcontrol}
\FDR(\AdaDetect_\alpha):=\E[\FDP(\AdaDetect_\alpha)] \leq  \alpha m_0/m.
\end{equation} 
Applying Corollary~\ref{ThresholdFDPboundDKW}, we obtain on the top of the in-expectation guarantee \eqref{fdrcontrol} the following uniform FDP bound for $\AdaDetect_\alpha$: with probability at least $1-\delta$, we have
\begin{align}
 &\forall \alpha\in (0,1),\:\:\FDP(\AdaDetect_\alpha)\leq \ol{\FDP}^{\mbox{\tiny DKW}}_{\alpha,\delta}\nonumber\\
&\ol{\FDP}^{\mbox{\tiny DKW}}_{\alpha,\delta} :=\paren{ \alpha\frac{\hat{m}_0}{m} + \frac{\hat{m}_0 \lambda^{\mbox{\tiny DKW}}_{\delta,n,\hat{m}_0}}{\hat{k}_\alpha\vee 1}}\ind{\hat{k}_\alpha>0},\label{AdaFDPboundDKW}
\end{align}
where $\hat{k}_\alpha$ is the rejection number \eqref{equkchapeau} of $\AdaDetect_\alpha$ and $\hat{m}_0$ satisfies \eqref{m0hat}.

In addition, we consider 
\begin{align}
 \ol{\FDP}^{\mbox{\tiny Simes}}_{\alpha,\delta} :=\frac{\hat{m}_0 \alpha}{m\delta }  \ind{\hat{k}_\alpha>0},\label{AdaFDPboundSimesb2}
\end{align}
for any estimator $\hat{m}_0$ given by \eqref{m0chapSimes}.

\section{Additional experiments}\label{sec:addexp}

In this section, we provide experiments to illustrate the FDP confidence bounds for AdaDetect, as mentioned in Remark~\ref{rem:FDPboundalpha} and Section~\ref{sec:AdaDetect}. 

The two procedures used are of the AdaDetect type \eqref{equkchapeau} but with two different score functions: the Random Forest classifier from \cite{marandon2022machine} (adaptive score), and the one class classifier Isolation Forest as in \citealp{bates2023testing} (non adaptive score).
The hyperparameters of these two machine learning algorithms are those given by \cite{marandonAdaImplementation}.

The FDP and the corresponding bounds are computed for the two procedures. The true discovery proportion is defined by
\begin{equation}\label{eq:TDP}
\TDP({R})=\frac{|R\cap \cH_1|}{|\cH_1|\vee 1},
\end{equation}
where $\cH_1= \range{m} \setminus \cH_0$;
this criterion will be considered in addition to the FDP to evaluate the detection power of the procedures.

Following the numerical experiments of \cite{marandon2022machine} and \cite{bates2023testing}, we consider the three different real data from OpenML dataset (CC-BY license)\citep{OpenML2013} given in Table~\ref{tab:datasets}.

\begin{table}
\caption{Summary of datasets. ``Shuttle'' is originally from UCI depository. ``Credit card" is from \cite{dal2015calibrating}. ``Mammography'' is from \cite{woods1993comparative}. \label{tab:datasets}}
\begin{center}
 \begin{tabular}{lcccc}
& Shuttle & Credit card  &  Mammography \\
\hline
Dimension $d$ & 9 & 30  & 6   \\
Feature type & Real & Real  & Real \\\hline
$|\mathcal{D}_{{\tiny \mbox{train}}}|$ & 3000 & 2000 &  2000  \\
$n$ calibration sample size & 2000  & 1000 & 1000  \\
$m_0$ (test) inlier number & 1500 & 500 & 500   \\
   $m_1$ (test) novelty number & 300 & 260& 260   \\
   $m=m_0 +m_1$ total test sample size & 1800 & 760 & 760 \\
\hline
\end{tabular}
\end{center}
\end{table} 

The results are displayed in Figure~\ref{fig:AdaDKW_alldata} for comparison of adaptive versus non-adaptive scores for the different FDP confidence bounds
and the TDP. On Figure~\ref{fig:AdaDKW_alldata2}, we focus on the adaptive scores and corresponding FDP bounds only; we compare the effect (on the bounds)
of demanding a more conservative error guarantee ($\delta=0.05$ versus $\delta=0.2$), as well as the effect of estimating
$m_0$ via \eqref{m0hat} instead of just using the inequality \eqref{ThresholdFDPboundDKW} with $\hat{m}_0 = m$.

The high-level conclusions are the following:
\begin{itemize}
\item using adaptive scores rather that non-adaptive ones results in a performance improvement (better true discovery proportion for the same
  target FDR level)
\item  for small target FDR level $\alpha$, the Simes upper bounds $\ol{\FDP}^{\tiny \mbox{Simes}}_{\alpha,\delta}$ are sharper than the DKW bound,
  elsewhere the new DKW bound is sharper than Simes.
  Furthermore, the relevant region for the Simes bound having the advantage becomes all the more tenuous as the error guarantee for the bound becomes more stringent
  (smaller $\delta$). The reason is that the Simes upper bound is linear in $\delta^{-1}$, while the DKW is only (square root) logarithmic.
\item estimating the estimator $\hat{m}_0$ from \eqref{m0hat} yields sharper bounds on the FDP and is therefore advantageous.
\end{itemize}

\begin{figure}[h!]
\begin{center}
\begin{tabular}{cc}
\rotatebox{90}{\hspace{3cm}Shuttle} &\includegraphics[scale=0.2]{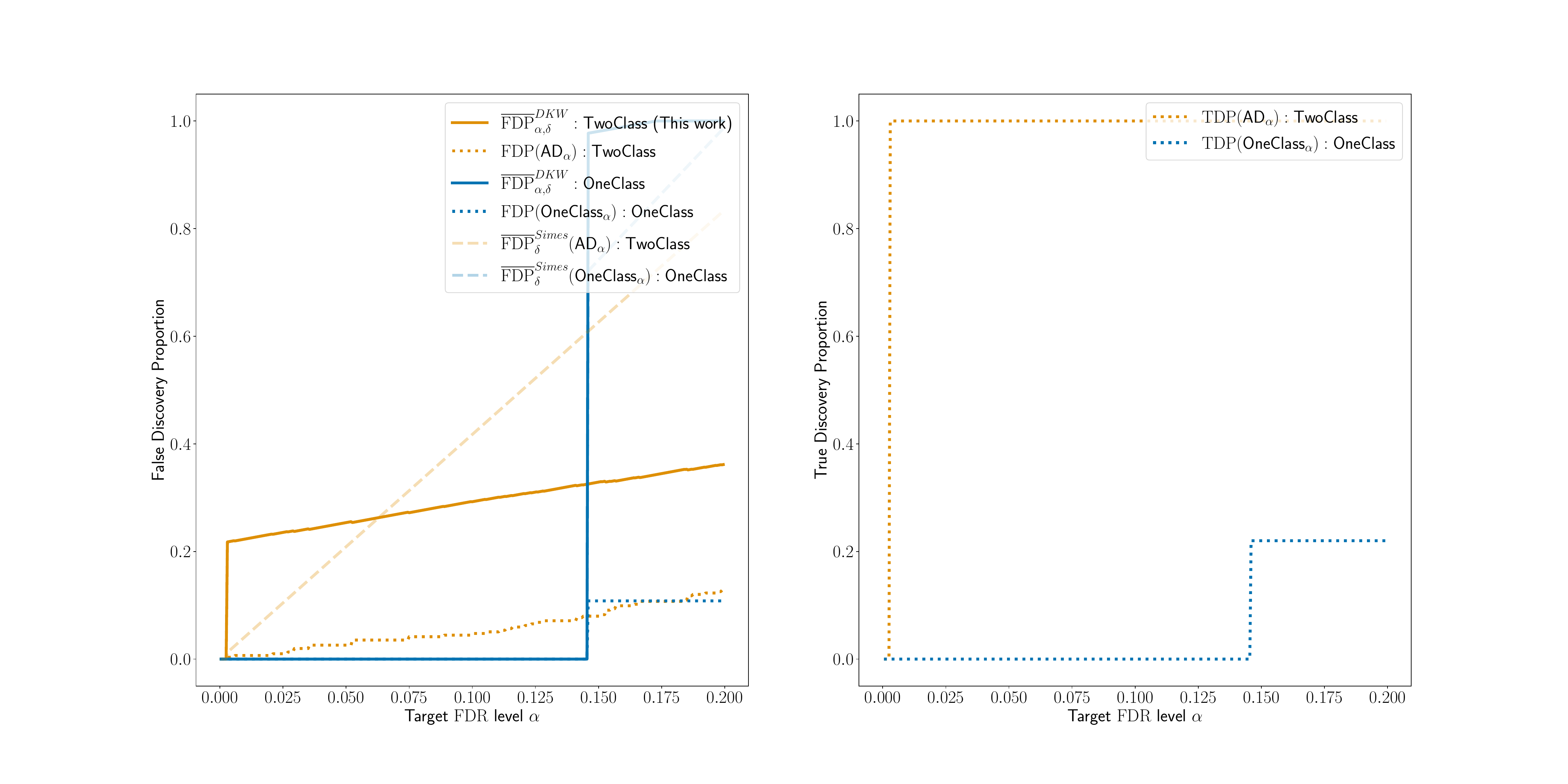}\\[-1cm]
\rotatebox{90}{\hspace{3cm}Credit Card} &\includegraphics[scale=0.2]{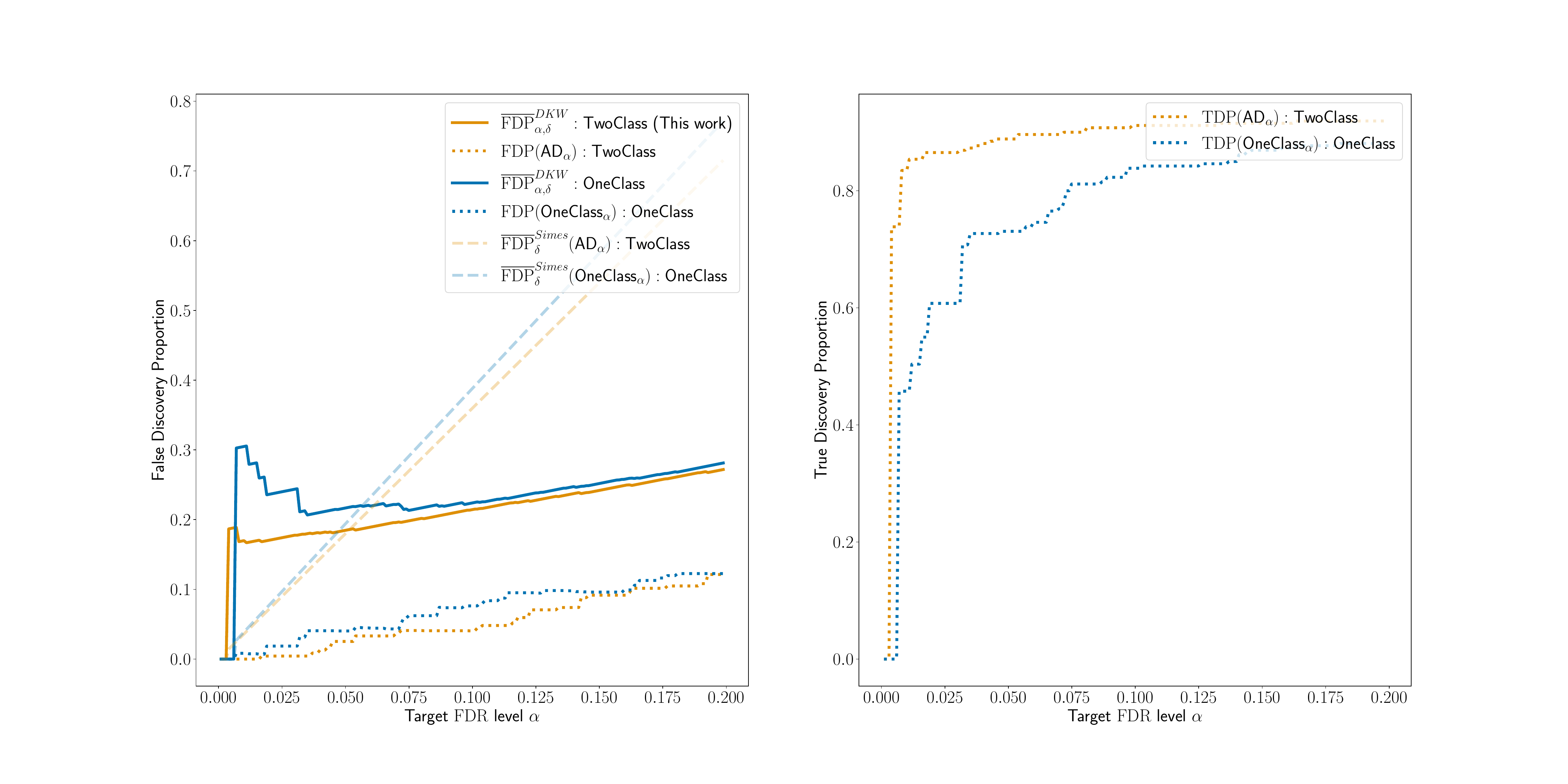}\\[-1cm]
\rotatebox{90}{\hspace{3cm}Mammography} &\includegraphics[scale=0.2]{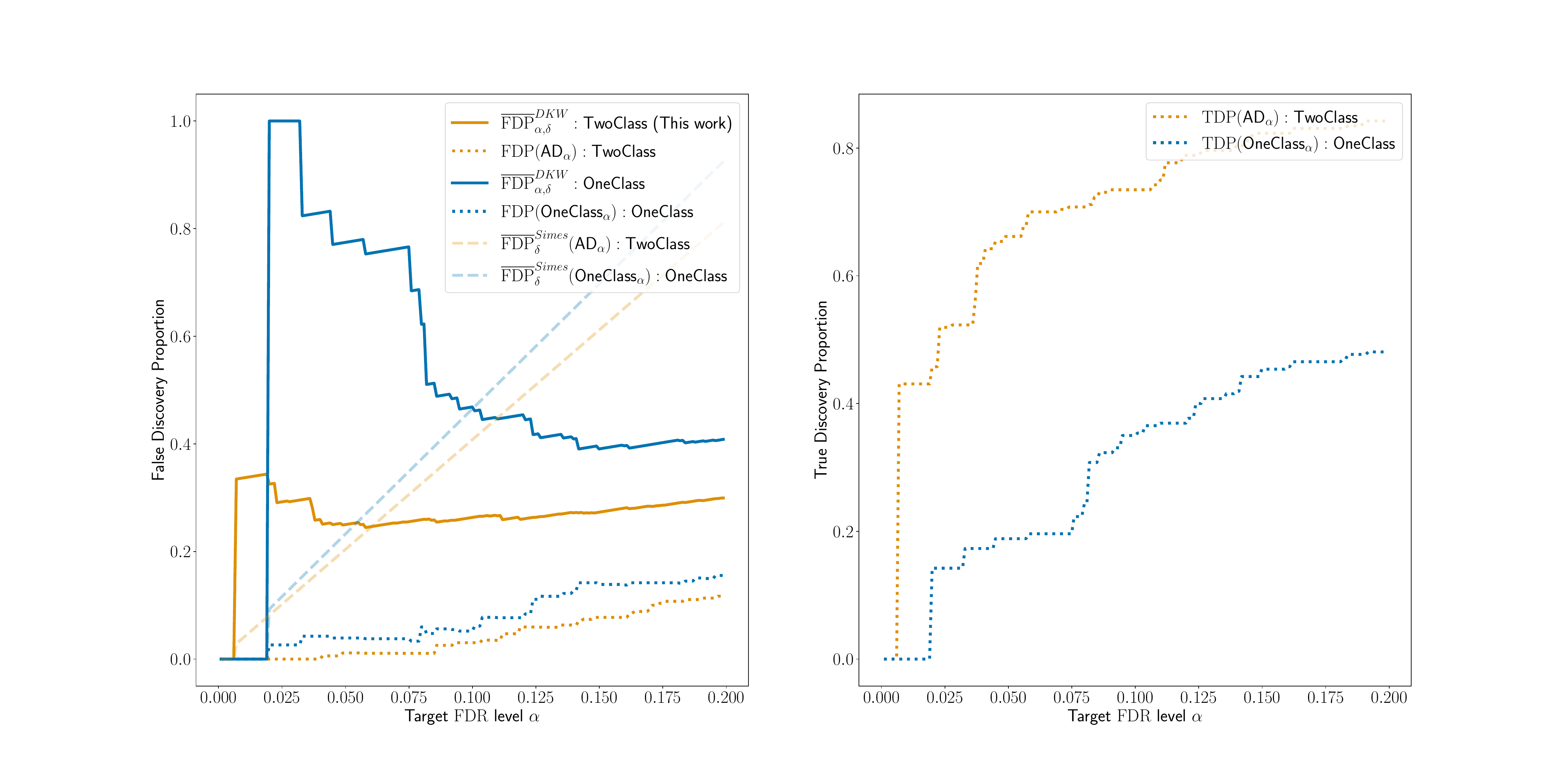}
\end{tabular}
\vspace{-1cm}
\end{center}
\caption{
Left: $\FDP(\AdaDetect_\alpha)$ \eqref{FDP}\eqref{equkchapeau} (dotted) and bounds ${\overline{\FDP}}^{\mbox{\tiny DKW}}_{\alpha,\delta}$ \eqref{AdaFDPboundDKW}  (solid)  $\ol{\FDP}^{\mbox{\tiny Simes}}_{\alpha,\delta}$ \eqref{AdaFDPboundSimesb2} (dashed) ($\delta=0.2$) in function of the nominal FDR-level $\alpha$. Right: corresponding $\TDP$ \eqref{eq:TDP}.  
In AdaDetect, the score is obtained either with a one-class classification (non-adaptive, blue) or a two-class classification (adaptive, orange); higher is better. 
\label{fig:AdaDKW_alldata}}
\end{figure}

\begin{figure}[h!]
\begin{center}
\begin{tabular}{ccc}
& Effect of $\delta$ & Effect of $\hat{m}_0$\\
\rotatebox{90}{\hspace{2.5cm}Shuttle} &\includegraphics[scale=0.25]{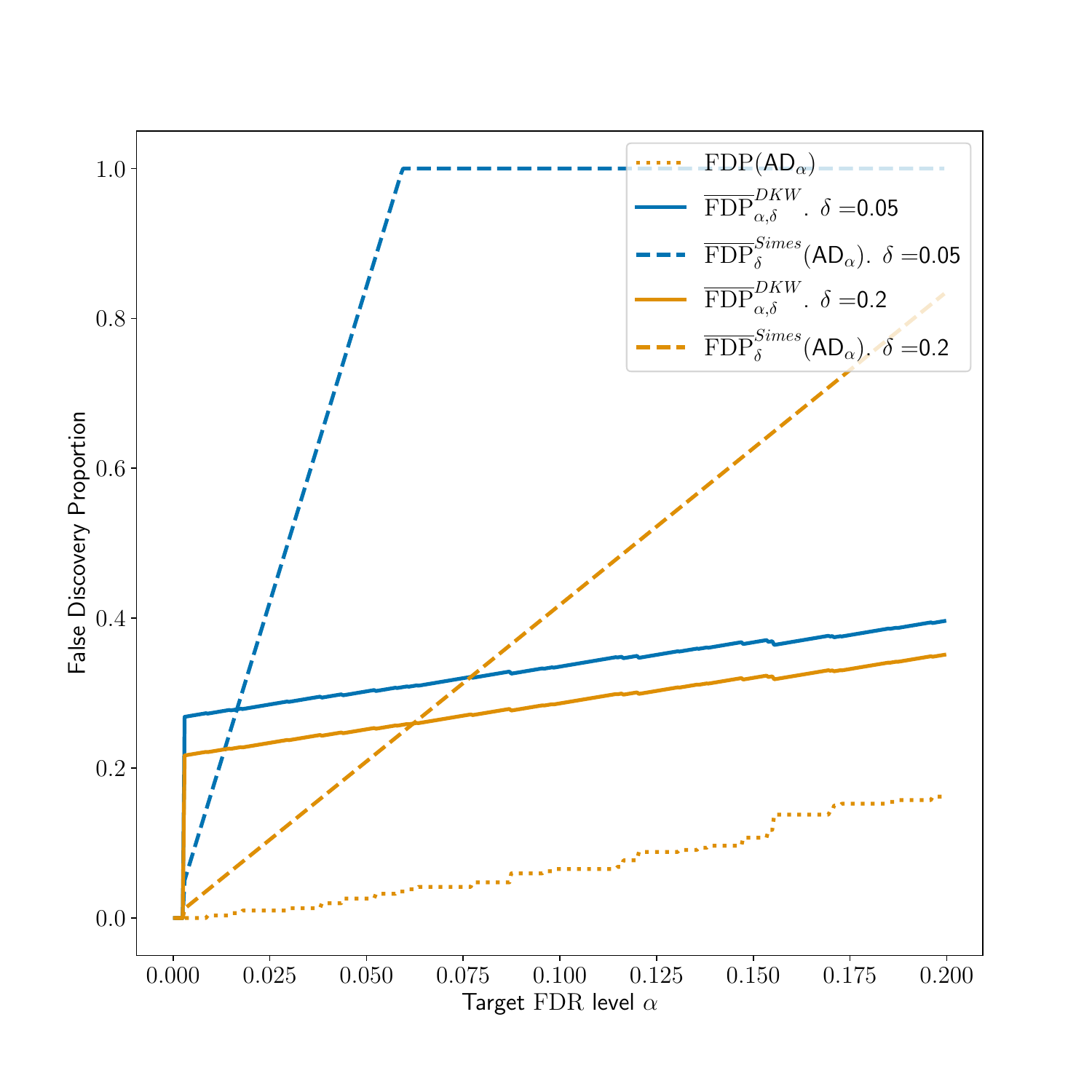}&\includegraphics[scale=0.25]{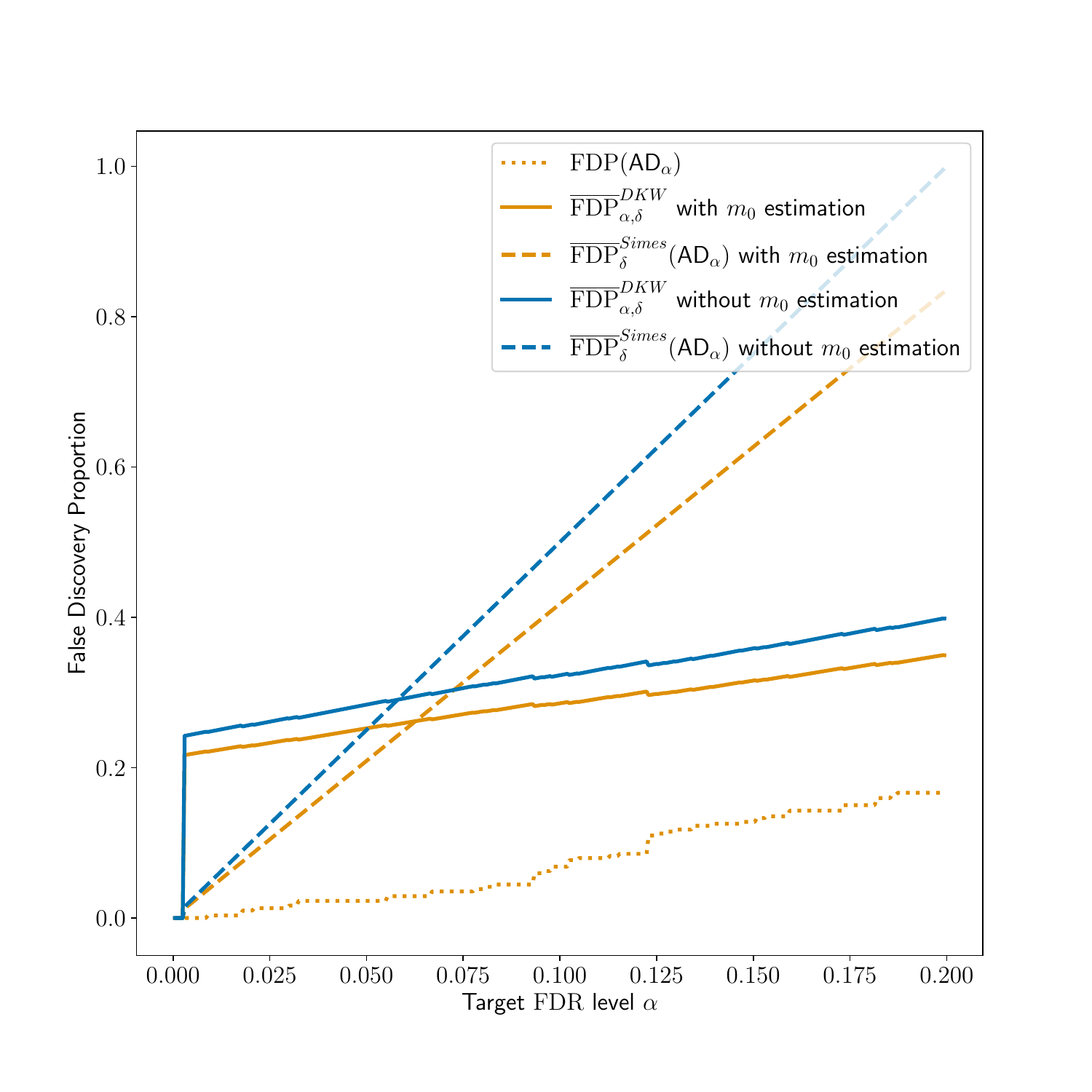}\\[-1cm]
\rotatebox{90}{\hspace{2.5cm}Credit Card} &\includegraphics[scale=0.25]{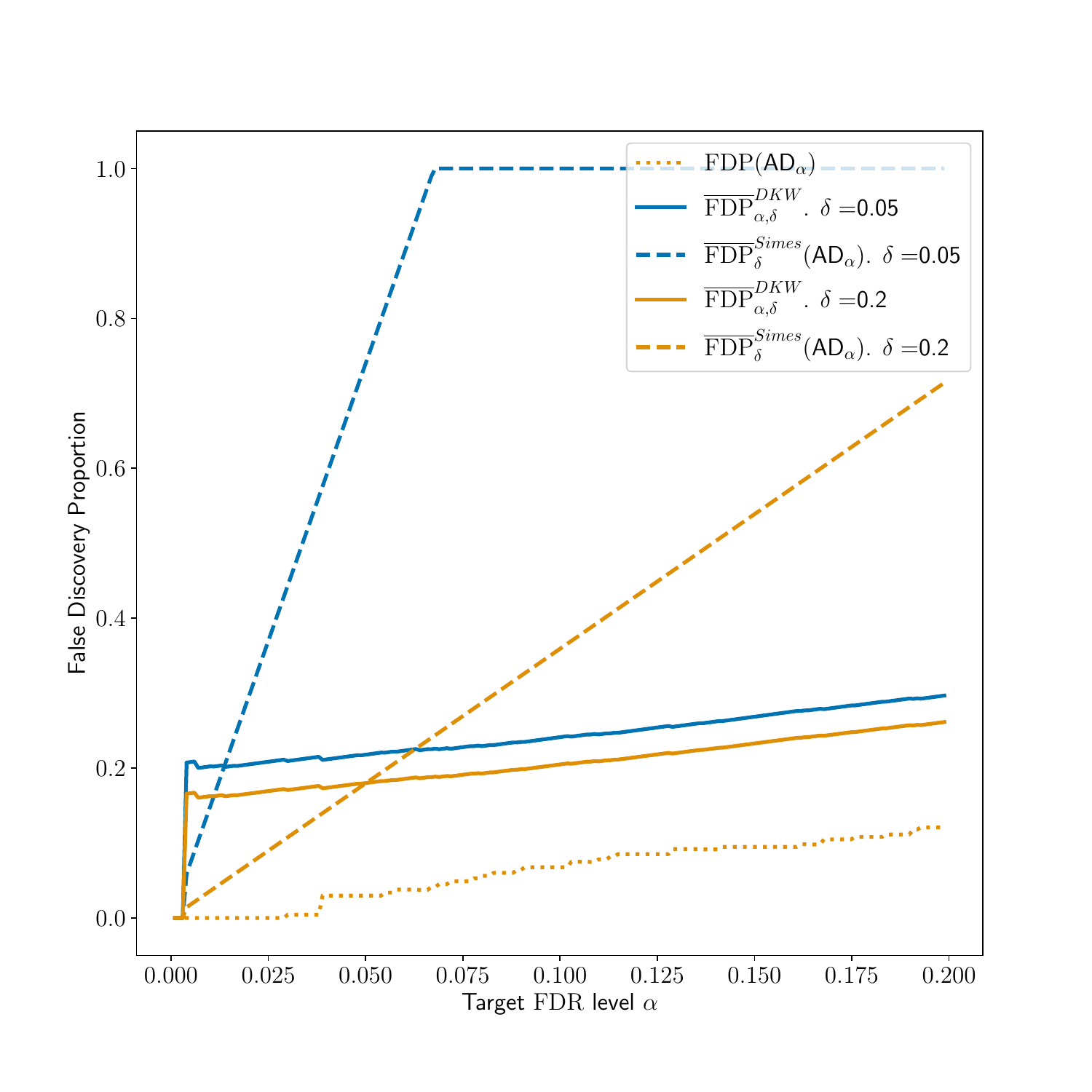}&\includegraphics[scale=0.25]{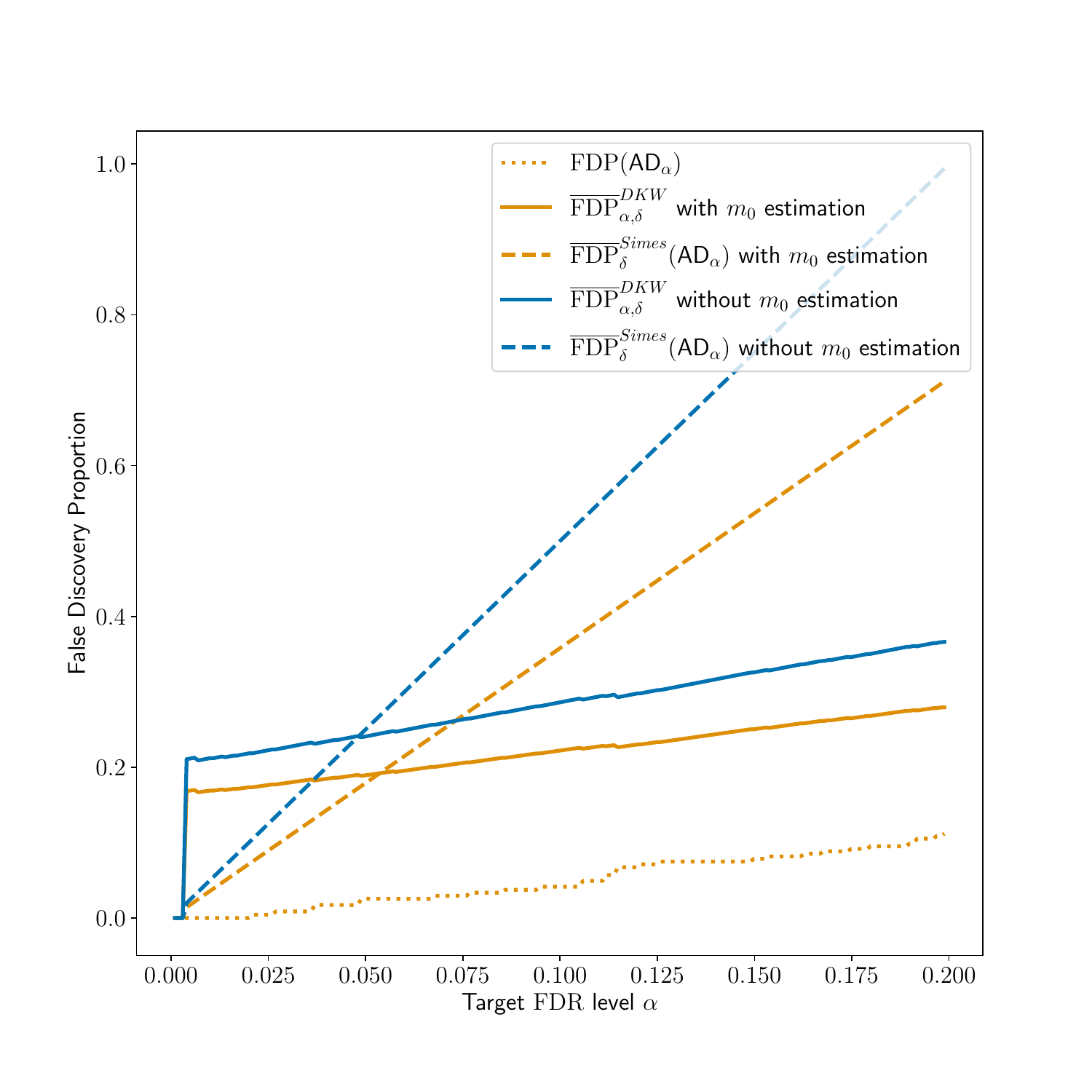}\\[-1cm]
\rotatebox{90}{\hspace{2.5cm}Mammography} &\includegraphics[scale=0.25]{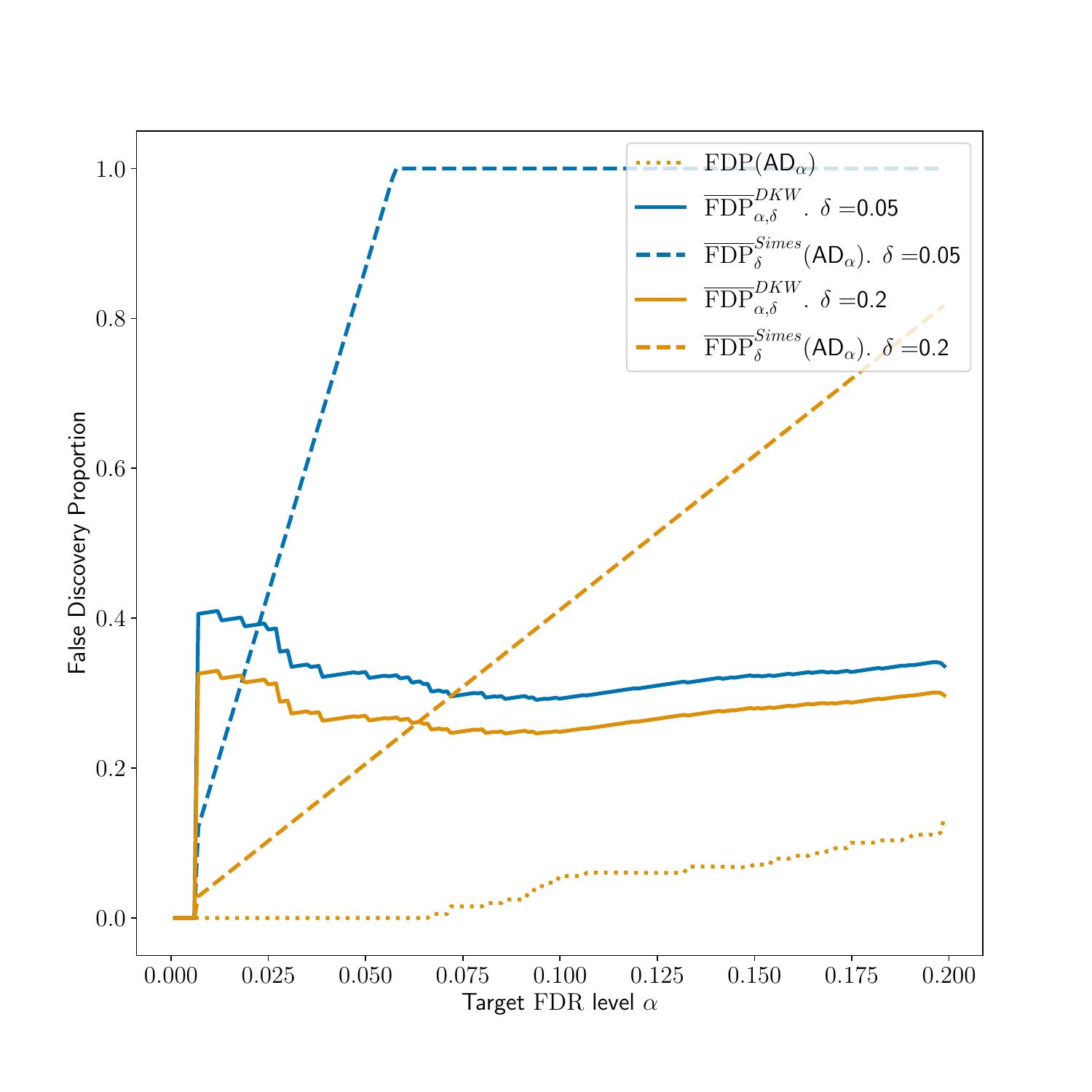}&\includegraphics[scale=0.25]{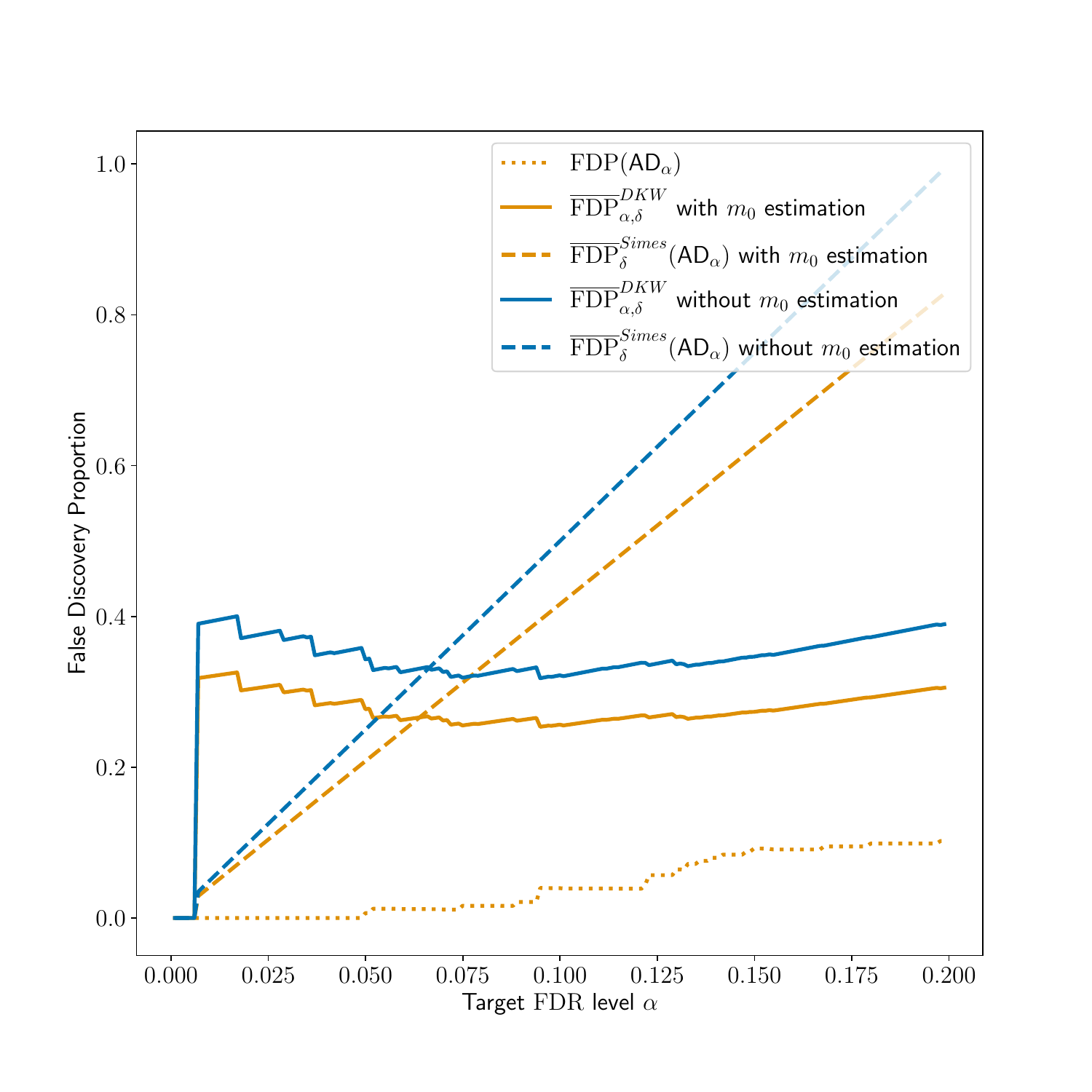}
\end{tabular}
\vspace{-1cm}
\end{center}
\caption{
  Same curves as Figure~\ref{fig:AdaDKW_alldata} (left), but only for two-class classification (adaptive, orange). Left: for comparison,
  the bounds ${\overline{\FDP}}$
  were also plotted for a smaller $\delta=0.05$ value (blue). Right: for comparison, bounds ${\overline{\FDP}}$ also plotted without an estimator of $m_0$
  (taking $m$ instead of $\hat{m}_0$).
\label{fig:AdaDKW_alldata2}}
\end{figure}

\end{document}